\documentclass[conference, onecolumn, compsoc, 11pt]{IEEEtran}

\usepackage{cite}
\usepackage[cmex10]{amsmath}
\usepackage{subfigure}
\usepackage{amsfonts}
\usepackage{amsthm}
\usepackage{epsfig}
\usepackage{graphicx}
\usepackage{algpseudocode}
\usepackage{algorithm}

\theoremstyle{remark} 
\theoremstyle{remark} 
\theoremstyle{plain} \newtheorem{ppt}{Property}
\theoremstyle{plain} 
\theoremstyle{plain} \newtheorem{prop}{Proposition}
\theoremstyle{plain} \newtheorem{thm}{Theorem}
\begin{document}
\setlength{\parskip}{0em}

\title{Deterministic Near-Optimal P2P Streaming}


\author{\IEEEauthorblockN{Shaileshh Bojja Venkatakrishnan\IEEEauthorrefmark{1}, Pramod Viswanath\IEEEauthorrefmark{2}} \IEEEauthorblockA{Coordinated Science Laboratory\\
University of Illinois - Urbana Champaign, USA\\
Email: \{\IEEEauthorrefmark{1}bjjvnkt2, \IEEEauthorrefmark{2}pramodv\}@illinois.edu} }

\begin{titlepage}
\maketitle
\thispagestyle{empty}
\end{titlepage}

\thispagestyle{empty}
\begin{abstract}
\normalfont
We consider streaming over a peer-to-peer network with homogeneous nodes in which a single source  broadcasts a data stream to all the users in the system. Peers are allowed to enter or leave the system (adversarially) arbitrarily. Previous approaches for streaming in this setting have either used randomized distribution graphs or structured trees with randomized maintenance algorithms. 
Randomized graphs handle peer churn well but have poor connectivity guarantees, while structured trees have good connectivity but have proven hard to maintain under peer churn. We improve upon both approaches by presenting a novel distribution structure with a {\em deterministic} and distributed algorithm for maintenance under peer churn; our result is inspired by a recent work \cite{haeupler2013simple} proposing deterministic algorithms for rumor spreading in graphs.  

A key innovation in our approach  is in having {\em redundant} links in the distribution structure. While this leads to a reduction in the maximum streaming rate possible, we show that for the amount of redundancy used, the delay guarantee of the proposed algorithm is near optimal. We introduce a {\em tolerance} parameter that captures the worst-case transient streaming rate received by the peers during churn events and characterize the fundamental tradeoff between rate, delay and tolerance. A natural generalization of the deterministic algorithm achieves this tradeoff near optimally. Finally,  the proposed deterministic algorithm is robust enough to handle various generalizations:  ability to deal with heterogeneous node capacities of the peers and more complicated streaming patterns where multiple source transmissions are present. 
\end{abstract}
\newpage

\section{Introduction} \label{sec: Introduction}
\setcounter{page}{1}
In peer-to-peer (P2P) streaming, a low-capacity server uploads the content to a small number of clients which, together with the other clients (a total of $n$ peers), then exchange the content among themselves. This is similar to the rumor spreading problem, in which a rumor from a source node is propagated to all the nodes of an unknown network. However unlike rumor spreading, where only a single rumor is communicated to neighbors over many rounds, in streaming new ``rumors'' arrive continuously in an online fashion and need to be forwarded fast and effectively in order to prevent message loss. Limited upload capacity of peers disallows flooding-type message forwarding. Further, peers can arrive or depart the system at will (peer churn), requiring scheduling algorithms to be designed in order to effectively utilize the upload capacity available and to ensure playback continuity with small delay. 

In this work, we consider the problem of constructing and maintaining a P2P overlay network $G(V,E)$ in a distributed fashion subject to the following restrictions. Peers can contact other peers, if they know their addresses, and form data carrying communication links on $G(V,E)$. While new addresses can be learnt by the peers by talking to their neighbors in $G(V,E)$, peers have a constant bound on (i) the number of addresses they can remember at any time, (ii) out-degree and (iii) upload capacity. Also, peers have only local knowledge of the topology of the graph $G(V,E)$. A server node receives data packets continuously as a live-stream from a source external to the network. The problem now is to construct $G(V,E)$ in order to distribute the data-stream from the source to all the peers. Additionally, we want to distribute the packets as quickly as possible (delay) and as many as possible in any given time duration (rate). As such there exist many algorithms that can stream optimally in this setting\cite{padmanabhan2001case, castro2003splitstream}. However, in practical P2P systems\cite{vu2006mapping}, peers seldom stay in the network all the time. To model this we let the nodes enter or exit the system arbitrarily. With this additional assumption on peer dynamics, repairing $G(V,E)$ to maintain fresh flow of packets to the peers, and ensuring good delay at the same time, becomes a challenging issue.  In particular, we consider a setting where the number of simultaneous connected departures is bounded (see section~\ref{sec: model}) but require that the peers may suffer a loss of at most a constant number of packets after each round of departure. For example, if any one peer departs the system then the remaining peers can experience a rate loss for at most a constant number of rounds before continuing to receive the full rate as before. 

A popular method used by some early systems, was to divide the content into multiple substreams and distribute via {\em multicast trees}  having disjoint interior nodes~\cite{padmanabhan2001case, castro2003splitstream, padmanabhan2003resilient, tran2003zigzag, zhang2012overlay}. This way any peer could be an interior node in one multicast tree where it utilizes its upload bandwidth. While trees offer good playback rate and delay, managing trees in a distributed fashion can be very difficult under peer churn. It is known that the complexity of maintaining trees grows with the number of nodes~\cite{liu2008performance, liu2010p2p}. Hence, random sampling by the peers has commonly been used to help maintain the distribution trees~\cite{zhu2013tree}. Another line of work introduced randomness in the construction of the distribution graphs themselves in order to handle the problems associated with peer churn~\cite{kim2013real}. Whenever a neighboring peer leaves, the peer chooses a new neighbor randomly as its new neighbor. While the distributed nature of the peer pairing makes unstructured networks robust to peer churn, connectivity is sacrificed because some of the peers may not be well connected due to the inherent randomness.

Thus, while structured algorithms promise connectivity to all the peers and have deterministic $O(\log n)$ delay guarantees, a fundamental limitation is their vulnerability to peer churn.  Randomized algorithms, on the other hand, provide only probabilistic guarantees for delay, convergence time and connectivity guarantees are weaker. Besides few algorithms (an exception is~\cite{kim2013real}) provide a formal guarantee on the transient rates received by the peers. A similar trend can be found in the literature on gossip, where a long line of algorithms tried to reduce the spreading time for randomized gossip~\cite{Doerr:2008:QRS:1347082.1347167, giakkoupis:LIPIcs:2011:2997, Censor-Hillel:2012:GCP:2213977.2214064}. However, recently a deterministic distributed algorithm for gossip was proposed in~\cite{haeupler2013simple}. Apart from being faster and more robust than previous randomized algorithms, the deterministic nature has the advantage of running time guarantees holding with certainty instead of with high probability. Inspired by this, we propose and analyze a novel distribution structure for P2P streaming that can be maintained deterministically, distributedly by the peers and provides a strong transient rate guarantee under our departure model. 

\subsection{~Our Results} 

The foremost challenge for the streaming problem in the setting discussed above is the design of the algorithm to construct and maintain the P2P overlay. This is in contrast to the literature on gossip wherein the model and algorithms, to some extent (such as uniform random gossip~\cite{ giakkoupis:LIPIcs:2011:2997, Censor-Hillel:2012:GCP:2213977.2214064}), are fairly standard and much of the innovation is in the analysis. Our main result is the design of the distribution structure and algorithm that is (i) distributed, (ii) deterministic and (iii) has constant repair time to ensure connectivity. As far as the authors are aware, no other algorithm in the literature has all of the above properties. The key innovation is the introduction of {\em redundancy} in the network. Assuming the peers have an upload capacity of $C$ each, the delay provided by our algorithm is given by the following.   

\begin{thm} \label{thm: steady state delay}
In the steady state if there are $n$ peers in the system, the streaming delay is bounded by $\log_2(n+1)+ \frac{2R}{C-R} + \log_2(1-\frac{R}{C})  -2$ for a rate $R\in(0,C)$.
\end{thm}
The above delay of our algorithm has an additional term of $O(1/(C-R))$ as compared to the $O(\log n)$ delay of tree based structures, such as in~\cite{castro2003splitstream,zhu2013tree}. However, the latter tree based algorithms do not have constant repair time under churn. Peer departures can cause a sudden loss of transmission links and can lead to loss of connectivity in the multicast graphs. Under such events, the data rate received by some peers can drop considerably until the trees are repaired. Having redundancy in the network helps in this regard in ensuring continuity of streaming without outages under peer churn. The penalties paid due to the introduction of redundant links facilitate:  (i) deterministic graph management and (ii) ensure continuity of playback under peer churn events. In our second result, we show that for the amount of redundancy used, the delay guarantee of the algorithm is order optimal.  

\begin{thm} \label{thm: order optimal delay}
For streaming using multiple structured graphs, each carrying partial flows, if each of the substream graphs have enough capacity redundancy to handle arbitrary node departures, then the maximum delay across the substream graphs is at least $\log_\Delta (n)+\frac{R}{2(C-R)}+\log_{\Delta}\left(\frac{2(C-R)}{R}\right)-c$, where $c =  (\Delta-2)\log_\Delta\left(\frac{\Delta!}{2}\right)   + \log_e(\Delta-1)+2$, for a rate $R$, degree bound $\Delta$ and $n\geq\frac{3R}{C-R}$.
\end{thm}

Thus, we claim that the $R/(C-R)$ term in the delay is {\em fundamental} for all algorithms guaranteeing continuity of playback. We also guarantee a transient rate equal to the original rate under bounded departure events. The transient rate can be traded for delay as discussed in section~\ref{sec: RDT tradeoff}. Hence apart from providing deterministic guarantees for delay and churn management, the algorithm offers key insights into the continuity aspect of the playback rate. It also extends readily for all-cast streaming, where every peer has a stream to be broadcast, and the case where the peers have heterogeneous upload capacities (Appendices~\ref{app: All-Cast} and~\ref{app: Heterogeneous nodes}). 

\subsection{~Related Work}

\begin{table}[t]
\normalsize
\renewcommand{\arraystretch}{1.1}
\centering
\begin{tabular}[b]{|c|c|c|}
\hline
Flow dissemination graph type & Graph maintenance algorithm & References \\ 
\hline \hline
 & Centralized & \cite{padmanabhan2003resilient, padmanabhan2001case, zhang2012overlay} \\ \cline{2-3} 
Structured & Involves randomness & \cite{castro2003splitstream, zhu2013tree, Stoica:2001:CSP:383059.383071, ratnasamy2001scalable, rowstron2001pastry, tran2003zigzag}\\ \cline{2-3}
&  Deterministic & \textbf{This paper} \\ 
\hline 
Unstructured & Random & \cite{mastwigka07,sanghavi2007gossiping,bonmasmatpertwi08, kim2013real, magharei2009prime, kostic2003bullet} \\
\hline
\end{tabular}
\caption{Summary of comparison with previous work for flow based streaming.}
\label{table: comparison}
\end{table}

A standard approach in structured streaming involved multicast trees, and often with constant-degree nodes \cite{mundinger2008optimal,kumar2007stochastic,liuchesenchilichou10}. Several approaches have been presented towards the management of the trees. Algorithms in~\cite{padmanabhan2003resilient, padmanabhan2001case, zhang2012overlay} used centralized control. Pastry~\cite{rowstron2001pastry}, a routing substrate, was used by the SplitStream algorithm in~\cite{castro2003splitstream} for tree construction and maintenance. Other distributed lookup protocols have also been proposed in~\cite{Stoica:2001:CSP:383059.383071, ratnasamy2001scalable}. In~\cite{zhu2013tree}, an asynchronous distributed algorithm was presented to construct and manage multiple distribution trees by means of random sampling done by the peers. In the other research direction of unstructured P2P networks, where each node communicates with a random subset of other peers, much of the previous theoretical studies of the delay performance have focused primarily on fully connected networks with homogeneous capacities; examples include \cite{mastwigka07,sanghavi2007gossiping,bonmasmatpertwi08} which make interesting connections between P2P streaming networks and gossip and epidemic models to analyze the maximum streaming delay. In~\cite{kim2013real}, multiple random Hamiltonian cycles were contructed and superposed. The distribution is then done over the union of the cycles. A key idea was to exploit the fact that the superposition of random directed Hamiltonian cycles is an expander with high probability. Additionally, Hamiltonian cycles are easy to maintain in response to peer churn, a fact that was first noted in the case of undirected graphs in \cite{lawsiu03}. Some other formats of unstructured P2P include mesh based streaming in~\cite{magharei2009prime, kostic2003bullet}, in which packets were distributed over a randomly constructed mesh. A comparison between the previous work discussed above and this work has been presented in Table~\ref{table: comparison}. We note that the idea of using redundancy to counter transient effects has been observed in other contexts as well~\cite{Alizadeh:2012:LMT:2228298.2228324, DBLP:journals/corr/BabarcziTRM14}.

\section{System Model} \label{sec: model}
The P2P overlay network $G(V(t),E(t))$, where the time $t$ is slotted, is assumed to be an undirected {\em node capacitated} graph in which all the peers have a uniform upload capacity of $C$ and a constant bound $\Delta$ on the number of upload links allowed. In addition to the upload capability, we let the nodes be able to communicate $O(\log n_{\mathrm{max}}(t))$ bits of information in any round $t$ as control messages through the edges where $n_{\mathrm{max}}(t)$ denotes the maximum number of peers that were in the system up until time $t$. Peers have a constant amount of memory for storing $M$ addresses (node ID's) and are allowed to arrive and depart from the system arbitrarily. When a peer departs, the node and all the edges connected to it are lost immediately; only the neighbors of the departing peer(s) in $G(V(t),E(t))$ are aware of this event. Let us call the maximally connected sets among the departing peers, in $G(V(t),E(t))$, as ``peer departure blocks''. For example, if $G(V(t),E(t))$ is as in Figure~\ref{fig:GenStruct_1} and peers 3, 5 and 6 leave the system at the same time, then \{3\}, \{5,6\} constitute the peer departure blocks. We assume that peer departure blocks, at any time, are of size at most $K$. Here $K$ has a linear dependence on $M$. Peer arrivals, in which a new peer becomes part of the overlay, happens at most one at a time. We also assume communication happens as a flow (or equivalently as time-shared discrete messages) and do not consider network coding in this paper. Notation: for any positive integer $n$, $[n]$ denotes the set $\{1,2,\ldots,n\}$. 

In the above model of peer dynamics, peers can potentially arrive or depart frequently. As such the P2P network can be constantly changing to adapt to that. Let us call such a state of the network, which is in the process of reconfiguring itself, as a {\em transient} state.   We call a state where the network is no longer changing as a {\em steady state}. This can happen, for example, if the time gap following a churn event and until the next churn event is large. In the following section, we discuss the steady state network structure of our algorithm.  

\section{Overview of Steady State Structure} \label{sec: Overview}
In the following sections, we let $C=1$ without loss of generality. We also let $K=1$ in the departure model. Let us consider a streaming of rate $R\leq1$ done over $T$ distribution graphs. Each distribution graph then is used to disseminate a rate $R/T$ substream to the peers. Let $G_i(V(t),E_i(t))$,  for $ i=1,\ldots,T$, denote the directed graph for broadcasting the $i$th substream, where $V(t)$ and $E_i(t)$ denote the set of all users in the system and the set of links used for the $i$th substream at time $t$ respectively. Each user is interested in receiving all the $m$ substreams (we do not assume any coding done over the data stream such as Multiple Description Coding (MDC)~\cite{castro2003splitstream}). For ease of notation, we will drop the argument $t$ from $G_i(t), E_i(t)$ etc. and denote them simply by $G_i, E_i$ etc., with the time aspect implicitly understood. Let us consider rational rates of the form $R=m/(m+1)$ for $m\in\mathbb{Z}^+$. Here we divide the stream into $T=m$ substreams of rate $1/(m+1)$ each. Any other general rate $R$ can be handled using $m=\lceil R/(1-R)\rceil$. 
\begin{figure}[t]
  \centerline{\subfigure[]{\includegraphics[height=35mm]{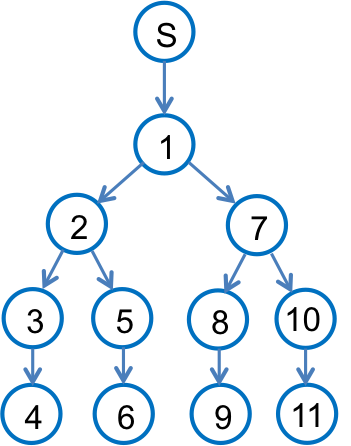}
     \label{fig:Parts_1}}
     \hfil
     \subfigure[]{\includegraphics[height=35mm]{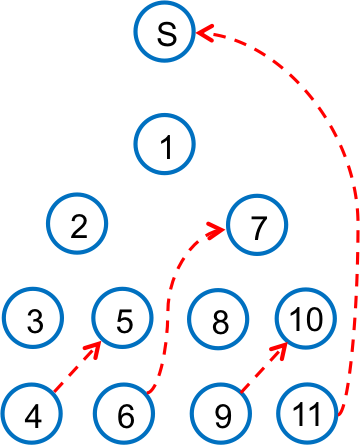}
     \label{fig:Parts_2}}
     \hfil
     \subfigure[]{\includegraphics[height=35mm]{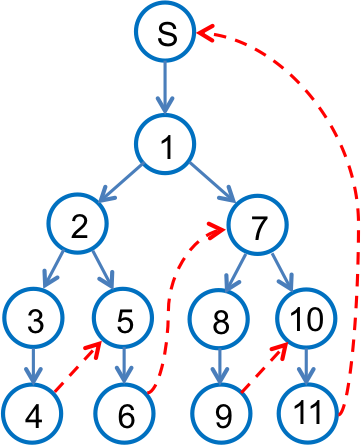}
    \label{fig:Parts_3}}}
    \caption{An example showing the directed graphs (a) $T_1$, (b) $U_1$, and (c) $G_1$  for a network with $n=11$ peers and $m=3$. The node indices correspond to the peer labels for $G_1$ and node S is the stream source.}
    \label{fig:Parts}
\end{figure}
Consider the steady state structure of $G_1$, i.e., after the graph has converged and when there is no more peer churn. Let there be $n$ nodes in the system in the steady state. Then, $G_1$ is the union of two graphs $T_1$ and $U_1$ described below. 

\textbf{Steady state:} $T_1$ is a directed binary tree with its root connected to the server and spans all the $n$ nodes. It is balanced in that for every degree two node, the size of the left subtree and the right subtree differ by at most one. We call the left outgoing edge as the primary edge and the right outgoing edge as the secondary edge. The corresponding children are called primary and secondary children respectively. The degree two nodes in $T_1$ are all close to the root of the tree, i.e., no degree one node  has a directed path leading to a degree two node. Further, the chain of degree one nodes leading to the leaf, for every leaf, consists of at least $m-1$ nodes and at most $2m-2$ nodes including the leaf node. In Figure~\ref{fig:Parts_1}, we have illustrated $T_1$ for $n=11$ and $m=3$. Now, given $T_1$, $U_1$ consists of edges that connect each leaf node of $T_1$ to the secondary child of the last degree 2 node in the path from the root to the leaf, such that, the secondary child itself does not lie in the path. For the $T_1$ shown in Figure~\ref{fig:Parts_1}, the graph $U_1$ has been illustrated in Figure~\ref{fig:Parts_2}. The steady state graph $G_1$ is the union of $T_1$ and $U_1$ and has been shown in Figure~\ref{fig:Parts_3}. Finally, the nodes in $G_1$ are labelled from the set of labels $\{1,\ldots,n\}$. The root node connected to the server has the label 1. For any degree two node $v$ in $T_1$ with label $l$ and a left subtree $L(v)$, its primary child has the label $l+1$ while the secondary child has the label $l+|L(v)|+1$ where $|L(v)|$ denotes the number of nodes in $L(v)$. This has also been shown in Figure~\ref{fig:Parts_3}. 

\begin{figure}[t]
  \centerline{\subfigure[]{\includegraphics[height=35mm]{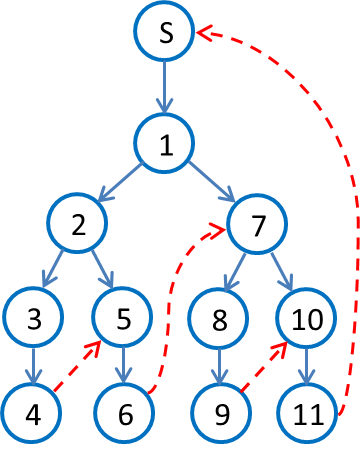}
     \label{fig:GenStruct_1}}
     \hfil
     \subfigure[]{\includegraphics[height=35mm]{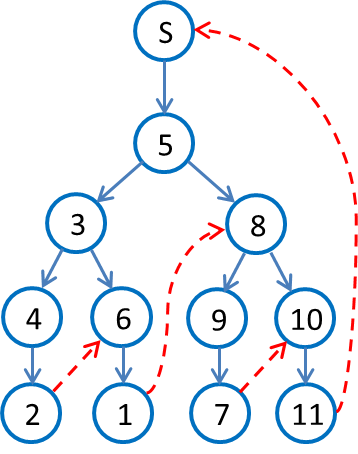}
     \label{fig:GenStruct_2}}
     \hfil
     \subfigure[]{\includegraphics[height=35mm]{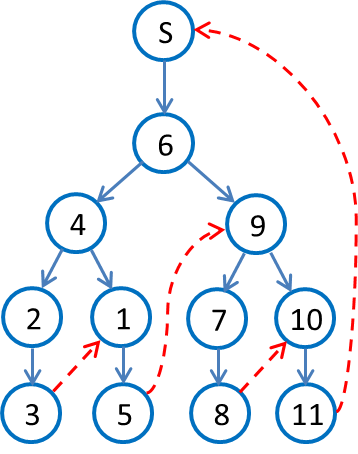}
    \label{fig:GenStruct_3}}}
    \caption{An example showing a steady state topology of $G_1, G_2$ and $G_3$ for $n=11$ and $R=3/4$ $(m=3)$. The node labels shown are with respect to the first substream.}
    \label{fig:GenStruct}
\end{figure}

The other substream graphs $G_2,\ldots,G_m$ also have a similar topological structure, but the peers with out-degree two in each of the $T_i$'s are different. This is illustrated in Figure~\ref{fig:GenStruct}. Each of the $G_i$ also has its own labeling similar to $G_1$, i.e., every peer has $m$ labels associated with it for the $m$ substreams. From the above discussion, it is easy to see the proof of Theorem~\ref{thm: steady state delay}. 

\textbf{Non steady-state:} Now, at any time instant not necessarily in the steady state, $T_i$ as a subgraph of $G_i$ is the shortest path graph for the peers from the server. The remaining edges form the edges in $U_i$. This redundancy in the form of edges in $U_i$ is critical for the algorithm to handle peer churn. For convenience, we have summarized key characteristics of the distribution graphs by the following properties.   
\begin{ppt} \label{ppt: connectivity}
For all $i\in[m]$, there exists a directed path from the server to all the peers in $G_i$. 
\end{ppt}
\begin{ppt} \label{ppt: structure}
For all $i\in[m]$, (i) every node has an out-degree of either one or two in $G_i$, (ii) every secondary child has an incoming edge from a leaf of $T_i$ and (iii) any node has an out-degree of two in at most one substream graph.  
\end{ppt}
\begin{ppt} \label{ppt: steady state structure}
In any steady state configuration, for all $i\in[m]$, we have (i) Properties~\ref{ppt: connectivity} and~\ref{ppt: structure} are satisfied, (ii) the sizes of the primary and secondary subtrees of any degree two node in $T_i$ differ by at most one, (iii) the chain of degree one nodes in $T_i$ have atleast $m-1$ nodes and at most $2m-2$ nodes including the leaf and (iv) no degree two peer has a degree one parent in $T_i$.  
\end{ppt}
These properties are used in the presentation of our algorithm in section~\ref{sec: Algorithm} and also in the subsequent sections. We now present the algorithm.  

\section{Algorithm} \label{sec: Algorithm}

The algorithm consists of a set of procedures that are run in parallel (and distributedly) by each peer. Since the algorithm is flow-based, deterministic and distributed, the message forwarding (peer and piece selection) itself is straight-forward and is discussed in section~\ref{sec: Label Control}. Sections~\ref{sec: Label update} and~\ref{sec: Peer Churn} discuss the operations that are required for the primary goal of ensuring connectivity and bandwidth whenever peers depart or arrive. The remaining sections~\ref{sec: Active Balance} and~\ref{sec: Induced Balance} deal with the secondary goal of balancing the topology in order to minimize the delay. For conciseness,  the detailed pseudo-code and illustrative representations of the procedures are moved to Appendix~\ref{app: Algorithm}.       

\subsection{~Label Control} \label{sec: Label Control}

Peer labels and addresses sent as control messages can be used by other peers to get an estimate of the sizes of their subtrees. This information is used by several procedures to follow, including the balancing subroutine. Each peer has to forward at most one (label, address) pair in each of the $G_i$. They cost at most $\log_2 n_\text{max}(t)$ bits per round which is minimal compared to the packet sizes. Suppose there are $n$ nodes in the system. The node labeled $n$ has an edge $(n,S)$ in $U_i$ where $S$ is the server. Hence, node 1 knows the maximum label index $n$ through $S$. Now, any degree two node in $T_i$ sends the label of its secondary child as the control message to its primary child. It sends the label received from its parent in $T_i$ as the control message to its primary child. A degree one node in $G_i$ simply forwards the control message received from its parent to its child. By sending control messages as above, any node with label $v$ in substream graph $G_i$ receiving a control message $l$ from its parent can know that the labels of the peers in its subtree ranges from $v+1$ to $l-1$. For example, in Figure~\ref{fig:GenStruct_1} node 2 receives the message $l=7$ since the label of the secondary child connected to node 1 is 7. As such the labels of the nodes in the subtree of node 2 range from 3 to 6.  

\subsection{~Label Update} \label{sec: Label update}
In the presence of peer churn, the labeling structure of section~\ref{sec: Overview} might no longer be valid -- some labels are no longer in the system, while others require new labels to be assigned to them. As such, a label update procedure constantly tries to keep the node labels updated in order to achieve the desired structure. In any substream graph the departure of a peer affects the labels of only those nodes which have a label greater than the label of the departed peer in that graph. For example in Figure~\ref{fig:GenStruct_1} the departure of node 4 will cause the labels of nodes 5 -- 11 to decrease by one.  As such the labels can be updated by broadcasting the label of the departed peer and the flag ``$-1$'', which essentially directs all nodes having a label value greater than the label of the departed peer to reduce their value by 1. This update can be performed quickly by using the edges of all $m$ substream graphs in order to do the global broadcast.

\subsection{~Peer Churn} \label{sec: Peer Churn}

An important characteristic of our present structure is the ease with which peer departures can be handled. Whenever a peer with (out) degree one in $G_i$ departs, a natural way to restore connectivity is for its child to connect to its parent. Further, the structure of the distribution graphs (Property~\ref{ppt: structure}) ensures that every secondary child of a degree two peer in $G_i$, receives an edge in $U_i$. As such, when a degree two peer departs, then the primary child connects to its parent, while the secondary child continues to receive the stream from the redundant edge in $U_i$. For example, in Figure~\ref{fig:GenStruct_1} (i) if peer 1 departs, the edge $(s,2)$ is formed; 7 receives the stream from 6; (ii) if peer 5 departs, the edges $(2,6),(4,6)$ are formed; (iii) if peer 6 departs, the edge $(5,7)$ is formed. Hence, the $G_i$'s continue to satisfy Property~\ref{ppt: structure} even under peer departures. 

The arrival procedure, whenever a new peer enters the system, should also be such that Property~\ref{ppt: structure} holds after the arrival. However, the main objective for any arriving peer is to first receive all the $m$ substreams. Whenever a new node arrives into the system it can contact an arbitrary node. The contacted node includes the new node as its child in all $G_i$'s where it has a degree one. In substreams where the new node has not yet been included (because the contacted node has a degree two), the new node can request to be the parent of one of its children from the previous substream trees. These operations preserve Properties~\ref{ppt: connectivity} and~\ref{ppt: structure}. The departure and arrival procedures have been illustrated in Figures~\ref{fig:RemNode} and~\ref{fig:NodeArr} respectively in Appendix~\ref{app: Algorithm}.

\subsection{~Active Balance} \label{sec: Active Balance}
This procedure is used as a sub-routine in the balancing algorithm of section~\ref{sec: Induced Balance}. Our balancing procedure is such that, even if only one of the trees $T_i$ is balanced, it can induce its topology onto the other substream trees in a cyclic fashion. The present procedure, Active Balance, is used only when none of the $T_i$'s are balanced. In this case, Active Balance tries to balance the first tree $T_1$, which can then balance the other trees. In this sense, it is used only as a last resort while balancing. 

From section~\ref{sec: Label Control} we know that peers can estimate the size(s) of the subtree(s) below them in each of the $T_i$'s by using the label messages received. If the labels have been updated, for any degree two node in $G_1$ with a label $v$ and incoming control message $l$ from its parent, if the label of its secondary child is not equal to $(v+l)/2$, then it is clear that the left and the right subtrees of $v$ are not balanced. As such, $v$ breaks its secondary edge and connects to the node with label $(v+l)/2$. Note that while $v$ knows that it needs to connect to the node labeled $(v+l)/2$, it might not know the physical address in order to initiate and complete the connection. One way to do this is for $v$ to request the physical address from the tracker server. Another alternative is to gossip the physical address of the desired node. The new links are also formed such that Property~\ref{ppt: structure} holds, i.e., whenever a peer with in-degree one in $G_i$ receives a new edge from a peer upstream, then the old edge becomes an edge in $U_i$ while the new edge becomes the primary receiving edge in $T_i$. If the in-degree of the receiving peer is two, then the primary edge is broken. 

\subsection{~Induced Balance} \label{sec: Induced Balance}

As mentioned in the previous section~\ref{sec: Active Balance}, Induced Balance is the primary balancing procedure of our algorithm and includes a collection of sub-routines.   
Let us assume that the graph $T_1$ is balanced. Then every leaf node in $T_1$ has an edge in $U_1$ going to a secondary child. We associate the degree two parent of such a secondary child with each of the degree one chain of nodes above the leaf. For example, in Figure~\ref{fig:GenStruct_1}, nodes \{3,4\}, \{5,6\} and \{8,9\} are associated with nodes 2, 1 and 7 respectively (the last set of peers \{10,11\} are atypical and are not associated with any degree two peer). Now, the way $G_2$ can be induced from $G_1$ is through a series of steps in which (i) the top-most node of the chain takes the place of the degree two node, (ii) the entire chain moves up by a node and (iii) the degree two node takes the place of the leaf node. For instance, in Figure~\ref{fig:GenStruct_2} (as induced by Figure~\ref{fig:GenStruct_1}), node 5 has taken the place of node 1, 6 has moved up and node 1 has taken the position of the leaf node 6. Implementing these three steps at all leaf nodes ensures that the resulting $G_2$ is structurally the same as $G_1$ but with a fresh set of degree two nodes. This can be done in a distributed fashion, since the peers in the degree one chain receive the address of the secondary child of the associated degree two node by the label forwarding procedure (section~\ref{sec: Label Control}).  

As such, whenever a tree $T_i$ is balanced, it tries to initiate the above three step procedure to induce its topology onto $T_{i+1}$ (modulo $m$, i.e., $T_{m+1}\equiv T_1$). If $T_{i+1}$ is already balanced, then such a request is turned down. Inducing the topology of $G_{i+1}$ from a balanced $G_i$ takes at most $5$ rounds in our algorithm (Appendix~\ref{app: Algorithm}). Therefore, if at least one of the $T_i$'s is balanced, then in at most $5m$ rounds, we expect all the trees to get balanced. If any degree two peer in $T_1$ is not balanced for $5m$ rounds, we initiate the Active Balance procedure in section~\ref{sec: Active Balance} in order to balance $T_1$. 

Also, the inducing procedure is initiated by the first node in the degree one chain such as nodes 3,5 or 8 in Figure~\ref{fig:GenStruct_1}. However, the request is made only if such nodes cannot already be a degree two node in $G_i$ -- if the degree one chain below a node is too long ($> 2m-2$), then a secondary edge is formed within $G_i$ itself. Similarly, if either of the subtrees of a degree two node contain less than $m-1$ nodes, then the secondary edge is broken. This ensures that (iii) of Property~\ref{ppt: steady state structure} holds. 

Thus, we have addressed the key issues of bandwidth management, connectivity and delay in the algorithm. In the following section, we give the proof of Theorem~\ref{thm: steady state delay}. The fast reconfiguration property and the linear scaling of $K$ with $M$ has been proved in Propostion~\ref{prop: small block} in the Appendix~\ref{app: Gen tree degrees}. 

\section{Proof of Theorem~\ref{thm: steady state delay}} \label{sec: Proofs}

\begin{prop} \label{prop: steady state structure}
In any steady state configuration Property~\ref{ppt: steady state structure} is satisfied. Conversely, any $G_1,\ldots,G_m$ satisfying Property~\ref{ppt: steady state structure} are stable states for the system. 
\end{prop}
\begin{proof}
Any peer that is completely disconnected from any of $G_1,\ldots,G_m$ can enter the system as a new peer by contacting the server. As such, in the steady state all peers are part of the substream graphs and satisfy Property~\ref{ppt: connectivity}. Properties~\ref{ppt: structure}--(i) and (iii) are locally enforced by the peers. Now, for any $G_1,\ldots,G_m$ satisfying Property~\ref{ppt: connectivity}, the forwarding of the label addresses by procedure Label Control (section~\ref{sec: Label Control}) makes sure that the leaf nodes of $T_1,\ldots,T_m$ connect to their corresponding secondary children. Hence Property~\ref{ppt: structure}--(ii) holds.  
Properties~\ref{ppt: steady state structure}--(ii), (iii) and (iv) follow because of the procedures in Induced Balance (section~\ref{sec: Induced Balance}). The balancing algorithm ensures that subtrees of every degree two peer is balanced. The supplementary procedures in Induced Balance also ensure that the degree one chains are between $m-1$ and $2m-2$ nodes long as discussed in section~\ref{sec: Induced Balance}. By the same procedure, if any degree one peer has a degree two child, then a new secondary edge is formed by the degree one peer since the subtree below it has to have larger than $2m-2$ peers. For the converse, consider any $G_1,\ldots,G_m$ satisfying Property~\ref{ppt: steady state structure}. The only procedures that change the topology of the graphs are in Induced Balance (section~\ref{sec: Induced Balance}). However, since the graphs are already balanced and the degree one chains have between $m-1$ and $2m-2$ nodes, neither the balancing algorithm nor the supplementary procedures change anything. 
\end{proof}

\begin{proof} [Proof of Theorem~\ref{thm: steady state delay}]
In the steady state, since Property~\ref{prop: steady state structure} holds, we have that the length of the degree one chains range from $m-1$ to $2m-2$. A balanced binary tree of depth $d$ has $2^{d-1}$ leaves and $2^d-1$ nodes. Therefore, we must have $ 2^{d-1}(m-1) +2^d-1 \leq n$,
\begin{align}
\Rightarrow d &\leq \log_2\left(\frac{n+1}{m+1} \right) \label{eq: steady state log depth} \\
\Rightarrow D &\leq  \log_2\left(\frac{n+1}{m+1} \right) + 2m-2 \\
&= \log_2(n+1) + \log_2(1-R) + \frac{2R}{1-R} - 2. \label{eq: steady state delay}
\end{align}
For a general upload capacity of $C$ instead of 1, by proportionately scaling the substream rates, we have the required delay bound. 
\end{proof}
In Section~\ref{sec: Gen tree degrees}, we show that the above delay of the algorithm is order optimal. We now briefly discuss the scenario of a lowered redundacy in the network. 
 
\section{Rate-Delay-Tolerance Tradeoff} \label{sec: RDT tradeoff}

Tree based algorithms, such as~\cite{castro2003splitstream, zhu2013tree}, have a delay guarantee of $\lceil \log_2 n \rceil$ for a streaming rate of $R=1$, while the algorithm we have presented has a weaker delay guarantee of order $O(\log n + 1/(1-R))$ (Theorem~\ref{thm: steady state delay}) for a rate $R\leq 1$ in steady state. This can be explained by introducing a parameter called tolerance, $\tau$. In the streaming algorithm discussed in sections~\ref{sec: Overview} -- \ref{sec: Proofs}, we had incorporated redundant capacity into the individual substream graphs using the edges in $U_i,i\in[m]$. Now, consider a scenario in which the the redundancy is reduced by a factor of $1-\tau$ for some $0\leq\tau\leq1$, i.e., let the edges in $U_i, \forall i\in [m]$, have a rate of $(1-\tau)/m$ instead of $1/m$ for $R=m/(m+1)$. The following proposition demonstrates the gain in the delay obtained for a lowered redundancy.
\begin{prop} \label{prop: tolerace tradeoff}
For a tolerance parameter $\tau$, the steady state delay guaranteed by the algorithm is bounded by
\begin{align}
D(R,\tau,n) \leq \log_2(n+1) - \log_2\left( \frac{R(1-\tau)}{1-R} +1\right) + \frac{2R(1-\tau)}{1-R} -2  ,\quad 0\leq\tau\leq 1, \label{eq: tolerace relation}
\end{align}
for $n$ peers in the system and a rate of $R$. 
\end{prop}
\begin{proof}
In the steady state of the original algorithm, the peers had a degree of one in all the substream graphs or they had a degree two in one of the graphs and degree one in all the rest. By a slight modification, we can make the algorithm more symmetric where every peer with degree two in some $T_i$ is necessarily a leaf node in some other tree in the steady state. This leads to a more even distribution of capacity, i.e., any peer has degree one in $m-2$ trees and degree zero, two in one tree each or it has degree one in all the trees in the steady state. This corresponds to a total upload capacity of $mr$ and $(m-2)r + r + r(1-\tau)$ respectively, where $r$ denotes the rate carried by each tree $T_i$. Therefore, we must have $mr + (1-\tau)r \leq 1  \Rightarrow r \leq \frac{1}{m+1-\tau}$. As such, in this scenario we can support a total rate of $R = m/(m+1-\tau)$ across the $m$ substream trees, which is higher than the rate $m/(m+1)$ of our algorithm. Since the topology is the same in both cases, by substituting for $m$ in Equation~\eqref{eq: steady state delay} for delay, we get the desired bound in Equation~\eqref{eq: tolerace relation}. 
\end{proof}

Proposition~\ref{prop: tolerace tradeoff} shows that for a rate of $R$, the steady state delay obtained by lowering the amount of redundancy in the system is lower. The extreme case in which there is no redundancy at all in the system, i.e. $\tau=1$, corresponds to tree based algorithms with a deterministic delay of $\lceil\log_2n\rceil$. Thus, we have obtained a relationship which shows the tradeoff between rate, delay and redundancy for the framework of our algorithm.   

For $\tau=0$, one implication of the way the substream graphs are structured (Property~\ref{ppt: structure}) is that connectivity of the nodes within the substream graphs (Property~\ref{ppt: connectivity}) directly translates to availability of download bandwidth from which peers can receive packets at a full rate of $R=m/(m+1)$. However, if we reduce the redundancy in the graphs, i.e., for $\tau > 0$, then with peer churn some of the peers have an upper bound of $(1-\tau)/(m+1)$ on the substream rates, even if the graphs are connected, until the graph stabilizes. It is important to note that, there is always enough capacity for the peers in the union of the substream graphs since every peer uploads at a rate at least as much as the download rate. The substream graphs essentially introduce an asymmetry in the distribution of the capacity of each node across the different substreams in order to reduce delay. The stabilization algorithm ensures that the excess capacity available in any substream graph is effectively transferred to those in need. However, for the duration of the stabilization, even with connectivity assumptions, we can only guarantee a rate of $(1-\tau)m/(m+1) = (1-\tau)R$ for the peers. This highlights the drawback with using a non-zero tolerance $\tau$; a large tolerance parameter can cause the transient drops in the rate received to be large. Hence, the lower rate and larger delay of our algorithm, compared to the tree based algorithms mentioned in the beginning of this section, has the advantage of  guaranteed continuous playback at full rate even during peer churn.  

\section{Converse} \label{sec: Gen tree degrees}

The streaming algorithm we have presented involved binary trees in the substream graphs.In general, the distribution graphs for streaming can be of any topology. However, in this section, we show that the steady state delay of our algorithm in Theorem~\ref{thm: steady state delay}, is order optimal within the general class of algorithms that use multiple arbitrarily structured graphs with redundancies for streaming. 

Consider a directed tree with $n$ nodes, where the nodes have out-degrees ranging from $0$ to $l$. Let $d^{(i)}$ denote the fraction of the nodes having an out-degree of $i$, for $i=0,\ldots,l$. It is clear that the tree with the lowest depth, for a given $(d^{(0)},\ldots,d^{(l)})$, has the largest degree nodes on the very top followed by the second largest degree nodes and so on. A lower bound for the depth of such a tree is given by (Proposition~\ref{prop: depth lower bound} in Appendix)  
\begin{align}
D \geq \frac{d^{(1)}}{d^{(0)}} + \log_{l}\left(1+\sum_{k=2}^l nd^{(k)}(k-1)\right) - (l-2)\log_l\left(\frac{l!}{2}\right)  .\label{eq: delay expression}
\end{align}
We now show that the delay in Theorem~\ref{thm: steady state delay} is order optimal among any algorithm satisfying the conditions of Theorem~\ref{thm: order optimal delay}. The full proof of Theorem~\ref{thm: order optimal delay} has been presented in Appendix~\ref{app: Gen tree degrees}. 
\begin{proof}[Proof sketch of Theorem~\ref{thm: order optimal delay}]
A general streaming algorithm can work over any connected graph of $n$ vertices (mesh), where each vertex has an out-degree of at most $\Delta$. In the steady state, if communication happens via flow (copy $+$ forward), and is deterministic, then one can always consider the flow to be an union of many constant rate sub-flows. Therefore, without loss of generality let us consider $T$ trees with the $i$th tree carrying a rate of $r_i$. The full topology of the multicast streams can include more edges than just the trees above. The trees simply correspond to the routes by which the packets arrive earliest from the source to the peers. Now, suppose any one node departs the system; then at least one or more of the trees are broken. As such, reception of flow at full rate is hindered for some of the nodes and needs to be restored as fast as possible, if not immediately. Restoring is possible only by contacting another node in the tree corresponding to the substream, that is still connected to the server. Here, we are looking at a class of algorithms in which such a restoration is done by means of redundant links. Within this class of algorithms (that are solutions to the problem) we have the converse result stated in the theorem. 

Let $d_i^{(j)}$ denote the fraction of nodes having an out-degree of $j$ in tree $i$. Clearly, 
\begin{align}
d_i^{(0)}+d_i^{(1)}+\ldots+d_i^{(l)} = 1, \quad \forall i=1,\ldots,T. \label{eq: frac sum}
\end{align}
Since any tree with $n$ nodes has $n-1$ edges, we have
\begin{align}
n(d_i^{(1)}+2d_i^{(2)}+\ldots+(l-1)d_i^{(l-1)}+ld_i^{(l)}) = n-1, \quad \forall i=1,\ldots,T. \label{eq: sum degree}
\end{align} 
Now, every degree $i$ node for $i\geq 2$ needs atleast $i-1$ redundant edges because of the capacity requirement in the theorem. As such, the cumulative node capacity constraint becomes
\begin{align}
\sum_{i=1}^T (n-1)r_i + n(d_i^{(2)}+2d_i^{(3)}+\ldots+(l-1)d_i^{(l)})r_i \leq n. \label{eq: tolerance}
\end{align}
The proof essentially obtains a lower bound for the expression in Equation~\eqref{eq: delay expression} based on above Equations~\eqref{eq: frac sum},~\eqref{eq: sum degree} and~\eqref{eq: tolerance}. The delay for the $i$-th tree $D_i$ can be lower bounded as
\begin{align}
D_i \geq \frac{1}{d_{i}^{(0)}} - (\log_e(l-1)+2)  + \log_{l}\left(1+nd^{(0)}_{i}\right) - (l-2)\log_l\left(\frac{l!}{2}\right),\quad\forall i=1,\ldots,T. \label{eq: delay lower bound deterministic}
\end{align}
The right hand side of the above is a decreasing function of $d_i^{(0)}$ in $(0,1)$. Equations~\eqref{eq: frac sum},~\eqref{eq: sum degree} and~\eqref{eq: tolerance} also yield  
\begin{align}
\min_i d_i^{(0)} \leq \frac{1}{R} - 1 + \frac{2}{n} \label{eq: d zero upper bound}
\end{align}
(the proofs for Equations~\eqref{eq: delay lower bound deterministic} and~\eqref{eq: d zero upper bound} have been discussed in Appendix~\ref{app: Gen tree degrees}). Letting $i^* = \mathrm{arg\,min}~ d_i^{(0)}$, the overall delay for the system can be bounded by the delay of the $i^*$-th tree: 
\begin{align}
\Rightarrow D\geq \log_ln+\frac{R}{2(1-R)}+\log_{l}\left(\frac{2(1-R)}{R}\right)- (l-2)\log_l\left(\frac{l!}{2}\right)   - \log_e(l-1)-2
\end{align}
for $n\geq 3R/(1-R)$. For $l=\Delta$ and a node capacity of $C$ (rather than 1) replacing $R$ by $R/C$, we get the desired theorem. Hence we can conclude that the steady state delay in our algorithm, Theorem~\eqref{thm: steady state delay}, is order optimal for the class of algorithms satisfying the property in the theorem.
\end{proof}

\section{Conclusion}

We have presented a deterministic algorithm for streaming over structured distribution graphs in a peer-to-peer network. The algorithm has the peer churn handling capability of unstructured algorithms combined with the deterministic delay guarantees of structured algorithms, thus offering the best of both worlds. We have also identified a tolerance parameter, that is related to the transient rate guarantee, and have discussed its relationship to  rate and delay. Continuity of streaming playback is an important quality of service metric that has been overlooked in the P2P streaming literature. For the class of algorithms we discussed, we have shown that an additional delay of $R/(C-R)$ is the price paid for ensuring continuity. In general, other forms of adding redundancy exist -- particularly coding techniques such as MDC or network coding. It would be interesting to study how these other methods interact with delay, rate and continuity. Implementing the present algorithm for practical real-world performance evaluation is also an important  future direction.   

\bibliographystyle{IEEEtran} 
\bibliography{mybib}

\begin{thebibliography}{10}
\providecommand{\url}[1]{#1}
\csname url@samestyle\endcsname
\providecommand{\newblock}{\relax}
\providecommand{\bibinfo}[2]{#2}
\providecommand{\BIBentrySTDinterwordspacing}{\spaceskip=0pt\relax}
\providecommand{\BIBentryALTinterwordstretchfactor}{4}
\providecommand{\BIBentryALTinterwordspacing}{\spaceskip=\fontdimen2\font plus
\BIBentryALTinterwordstretchfactor\fontdimen3\font minus
  \fontdimen4\font\relax}
\providecommand{\BIBforeignlanguage}[2]{{%
\expandafter\ifx\csname l@#1\endcsname\relax
\typeout{** WARNING: IEEEtran.bst: No hyphenation pattern has been}%
\typeout{** loaded for the language `#1'. Using the pattern for}%
\typeout{** the default language instead.}%
\else
\language=\csname l@#1\endcsname
\fi
#2}}
\providecommand{\BIBdecl}{\relax}
\BIBdecl

\bibitem{haeupler2013simple}
B.~Haeupler, ``Simple, fast and deterministic gossip and rumor spreading.'' in
  \emph{SODA}.\hskip 1em plus 0.5em minus 0.4em\relax SIAM, 2013, pp. 705--716.

\bibitem{padmanabhan2001case}
V.~N. Padmanabhan and K.~Sripanidkulchai, ``The case for cooperative
  networking,'' in \emph{Revised Papers from the First International Workshop
  on Peer-to-Peer Systems IPTPS}, 2001.

\bibitem{castro2003splitstream}
M.~Castro, P.~Druschel, A.-M. Kermarrec, A.~Nandi, A.~Rowstron, and A.~Singh,
  ``Splitstream: High-bandwidth content distribution in cooperative
  environments,'' in \emph{Peer-to-Peer Systems II}.\hskip 1em plus 0.5em minus
  0.4em\relax Springer, 2003, pp. 292--303.

\bibitem{vu2006mapping}
L.~Vu, I.~Gupta, J.~Liang, and K.~Nahrstedt, ``Mapping the pplive network:
  Studying the impacts of media streaming on p2p overlays,'' 2006.

\bibitem{padmanabhan2003resilient}
V.~N. Padmanabhan, H.~J. Wang, and P.~A. Chou, ``Resilient peer-to-peer
  streaming,'' in \emph{Network Protocols, 2003. Proceedings. 11th IEEE
  International Conference on}.\hskip 1em plus 0.5em minus 0.4em\relax IEEE,
  2003, pp. 16--27.

\bibitem{tran2003zigzag}
D.~A. Tran, K.~A. Hua, and T.~Do, ``Zigzag: An efficient peer-to-peer scheme
  for media streaming,'' in \emph{INFOCOM 2003. Twenty-Second Annual Joint
  Conference of the IEEE Computer and Communications. IEEE Societies},
  vol.~2.\hskip 1em plus 0.5em minus 0.4em\relax IEEE, 2003, pp. 1283--1292.

\bibitem{zhang2012overlay}
W.~Zhang, Q.~Zheng, H.~Li, and F.~Tian, ``An overlay multicast protocol for
  live streaming and delay-guaranteed interactive media,'' \emph{Journal of
  Network and Computer Applications}, vol.~35, no.~1, pp. 20--28, 2012.

\bibitem{liu2008performance}
S.~Liu, R.~Zhang-Shen, W.~Jiang, J.~Rexford, and M.~Chiang, ``Performance
  bounds for peer-assisted live streaming,'' in \emph{ACM SIGMETRICS
  Performance Evaluation Review}, vol.~36, no.~1.\hskip 1em plus 0.5em minus
  0.4em\relax ACM, 2008, pp. 313--324.

\bibitem{liu2010p2p}
S.~Liu, M.~Chen, S.~Sengupta, M.~Chiang, J.~Li, and P.~A. Chou, ``P2p streaming
  capacity under node degree bound,'' in \emph{Distributed Computing Systems
  (ICDCS), 2010 IEEE 30th International Conference on}.\hskip 1em plus 0.5em
  minus 0.4em\relax IEEE, 2010, pp. 587--598.

\bibitem{zhu2013tree}
J.~Zhu and B.~Hajek, ``Tree dynamics for peer-to-peer streaming,'' \emph{arXiv
  preprint arXiv:1308.1971}, 2013.

\bibitem{kim2013real}
J.~Kim and R.~Srikant, ``Real-time peer-to-peer streaming over multiple random
  hamiltonian cycles,'' \emph{Information Theory, IEEE Transactions on},
  vol.~59, no.~9, pp. 5763--5778, 2013.

\bibitem{Doerr:2008:QRS:1347082.1347167}
\BIBentryALTinterwordspacing
B.~Doerr, T.~Friedrich, and T.~Sauerwald, ``Quasirandom rumor spreading,'' in
  \emph{Proceedings of the Nineteenth Annual ACM-SIAM Symposium on Discrete
  Algorithms}, ser. SODA '08.\hskip 1em plus 0.5em minus 0.4em\relax
  Philadelphia, PA, USA: Society for Industrial and Applied Mathematics, 2008,
  pp. 773--781. [Online]. Available:
  \url{http://dl.acm.org/citation.cfm?id=1347082.1347167}
\BIBentrySTDinterwordspacing

\bibitem{giakkoupis:LIPIcs:2011:2997}
\BIBentryALTinterwordspacing
G.~Giakkoupis, ``{Tight bounds for rumor spreading in graphs of a given
  conductance},'' in \emph{28th International Symposium on Theoretical Aspects
  of Computer Science (STACS 2011)}, ser. Leibniz International Proceedings in
  Informatics (LIPIcs), T.~Schwentick and C.~D{\"u}rr, Eds., vol.~9.\hskip 1em
  plus 0.5em minus 0.4em\relax Dagstuhl, Germany: Schloss
  Dagstuhl--Leibniz-Zentrum fuer Informatik, 2011, pp. 57--68. [Online].
  Available: \url{http://drops.dagstuhl.de/opus/volltexte/2011/2997}
\BIBentrySTDinterwordspacing

\bibitem{Censor-Hillel:2012:GCP:2213977.2214064}
\BIBentryALTinterwordspacing
K.~Censor-Hillel, B.~Haeupler, J.~Kelner, and P.~Maymounkov, ``Global
  computation in a poorly connected world: Fast rumor spreading with no
  dependence on conductance,'' in \emph{Proceedings of the Forty-fourth Annual
  ACM Symposium on Theory of Computing}, ser. STOC '12.\hskip 1em plus 0.5em
  minus 0.4em\relax New York, NY, USA: ACM, 2012, pp. 961--970. [Online].
  Available: \url{http://doi.acm.org/10.1145/2213977.2214064}
\BIBentrySTDinterwordspacing

\bibitem{Stoica:2001:CSP:383059.383071}
\BIBentryALTinterwordspacing
I.~Stoica, R.~Morris, D.~Karger, M.~F. Kaashoek, and H.~Balakrishnan, ``Chord:
  A scalable peer-to-peer lookup service for internet applications,'' in
  \emph{Proceedings of the 2001 Conference on Applications, Technologies,
  Architectures, and Protocols for Computer Communications}, ser. SIGCOMM
  '01.\hskip 1em plus 0.5em minus 0.4em\relax New York, NY, USA: ACM, 2001, pp.
  149--160. [Online]. Available: \url{http://doi.acm.org/10.1145/383059.383071}
\BIBentrySTDinterwordspacing

\bibitem{ratnasamy2001scalable}
S.~Ratnasamy, P.~Francis, M.~Handley, R.~Karp, and S.~Shenker, \emph{A scalable
  content-addressable network}.\hskip 1em plus 0.5em minus 0.4em\relax ACM,
  2001, vol.~31, no.~4.

\bibitem{rowstron2001pastry}
A.~Rowstron and P.~Druschel, ``Pastry: Scalable, decentralized object location,
  and routing for large-scale peer-to-peer systems,'' in \emph{Middleware
  2001}.\hskip 1em plus 0.5em minus 0.4em\relax Springer, 2001, pp. 329--350.

\bibitem{mastwigka07}
L.~Massoulie, A.~Twigg, C.~Gkantsidis, and P.~Rodriguez, ``Randomized
  decentralized broadcasting algorithms,'' in \emph{Proceedings of IEEE
  INFOCOM}, 2007.

\bibitem{sanghavi2007gossiping}
S.~Sanghavi, B.~Hajek, and L.~Massoulie, ``Gossiping with multiple messages,''
  \emph{IEEE Transactions on Information Theory}, 2007.

\bibitem{bonmasmatpertwi08}
T.~Bonald, L.~Massoulie, F.~Mathieu, D.~Perino, and A.~Twigg, ``Epidemic live
  streaming: Optimal performance trade-offs,'' in \emph{Proceedings of ACM
  SIGMETRICS}, Annapolis, MD, June 2008.

\bibitem{magharei2009prime}
N.~Magharei and R.~Rejaie, ``Prime: Peer-to-peer receiver-driven mesh-based
  streaming,'' \emph{IEEE/ACM Transactions on Networking (TON)}, vol.~17,
  no.~4, pp. 1052--1065, 2009.

\bibitem{kostic2003bullet}
D.~Kosti{\'c}, A.~Rodriguez, J.~Albrecht, and A.~Vahdat, ``Bullet: High
  bandwidth data dissemination using an overlay mesh,'' in \emph{ACM SIGOPS
  Operating Systems Review}, vol.~37, no.~5.\hskip 1em plus 0.5em minus
  0.4em\relax ACM, 2003, pp. 282--297.

\bibitem{mundinger2008optimal}
J.~Mundinger, R.~Weber, and G.~Weiss, ``Optimal scheduling of peer-to-peer file
  dissemination,'' \emph{Journal of Scheduling}, vol.~11, no.~2, pp. 105--120,
  2008.

\bibitem{kumar2007stochastic}
R.~Kumar, Y.~Liu, and K.~Ross, ``Stochastic fluid theory for p2p streaming
  systems,'' in \emph{INFOCOM 2007. 26th IEEE International Conference on
  Computer Communications. IEEE}.\hskip 1em plus 0.5em minus 0.4em\relax IEEE,
  2007, pp. 919--927.

\bibitem{liuchesenchilichou10}
S.~Liu, M.~Chen, S.~Sengupta, M.~Chiang, J.~Li, and P.~A. Chou, ``{P2P}
  streaming capacity under node degree bound,'' in \emph{Proceedings of IEEE
  ICDCS}, 2010.

\bibitem{lawsiu03}
C.~Law and K.-Y. Siu, ``Distributed construction of random expander networks,''
  in \emph{Proc. IEEE INFOCOM}, 2003.

\bibitem{Alizadeh:2012:LMT:2228298.2228324}
\BIBentryALTinterwordspacing
M.~Alizadeh, A.~Kabbani, T.~Edsall, B.~Prabhakar, A.~Vahdat, and M.~Yasuda,
  ``Less is more: Trading a little bandwidth for ultra-low latency in the data
  center,'' in \emph{Proceedings of the 9th USENIX Conference on Networked
  Systems Design and Implementation}, ser. NSDI'12.\hskip 1em plus 0.5em minus
  0.4em\relax Berkeley, CA, USA: USENIX Association, 2012, pp. 19--19.
  [Online]. Available: \url{http://dl.acm.org/citation.cfm?id=2228298.2228324}
\BIBentrySTDinterwordspacing

\bibitem{DBLP:journals/corr/BabarcziTRM14}
P.~Babarczi, J.~Tapolcai, L.~R{\'o}nyai, and M.~M{\'e}dard, ``Resilient flow
  decomposition of unicast connections with network coding,'' \emph{CoRR}, vol.
  abs/1401.6670, 2014.

\bibitem{munwebwei06a}
J.~Mundinger, R.~Weber, and G.~Weiss, ``Analysis of peer-to-peer file
  dissemination amongst users of different upload capacities,''
  \emph{Performance Evaluation Review, Performance 2005 Issue}, 2006.

\end{thebibliography}

\appendices

\section{Algorithm} \label{app: Algorithm}

In the following we have presented the pseudo-code for the procedures discussed in sections~\ref{sec: Label Control}--~\ref{sec: Induced Balance}. Illustrations have also been included.  

\subsection{~Label Control} 
\begin{algorithm}[H]
\caption{Message Forwarding algorithm for node $v$ in $G_i$}\label{forward}
\begin{algorithmic}[1]
\Require degree of node in $T_i$ and the addresses, labels of children;
\Procedure{Forward}{msg, $l$, add}\Comment{msg: streaming message; $l$: label; add: address}
\If{$\textrm{degree}=2$} 
\State send (msg, secondary child's label and address) to primary child in $T_i$; 
\State send (msg, $l$, add) to secondary child in $T_i$; 
\ElsIf{$\textrm{degree}=1$}
\State send (msg, $l$, add) to child in $T_i$;
\Else
\State send (msg) to child in $U_i$;
\EndIf
\EndProcedure
\end{algorithmic}
\end{algorithm}

\subsection{~Label Update}
\begin{algorithm}[H]
\caption{Label Update Algorithm for node $v$ in all substream graphs}\label{LabelUpdate}
\begin{algorithmic}[1]
\Require label of parent node(s) of $v$ in $G_i,i=1,\ldots,m$;
\Procedure{Update}{$i,l,f,t$}\Comment{$l$: label of departed/arrived node in tree $i$ at time $t$; $f$: flag}
\If{this label update not already done \textbf{and} label($v$) $\geq l_i$}\Comment{check using time-stamp $t$}
\State label($v$) $\gets l_i + f_i$ 
\EndIf 
\State forward $(i,l,f,t)$ to all edges (undirected) other than the received edge in $T_1\cup T_2\cup \ldots \cup T_m$;
\EndProcedure
\end{algorithmic}
\end{algorithm}

\subsection{~Peer Churn}

\begin{figure}[t]
  \centerline{\subfigure[]{\includegraphics[height=35mm]{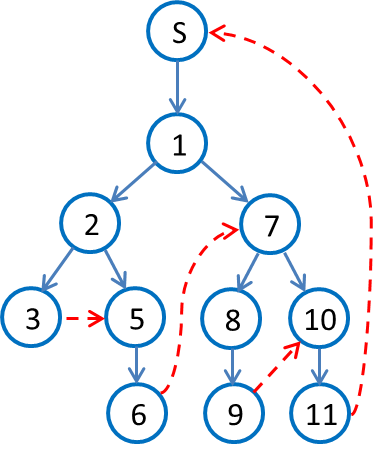}
     \label{fig:RemNode_1}}
     \hfil
     \subfigure[]{\includegraphics[height=35mm]{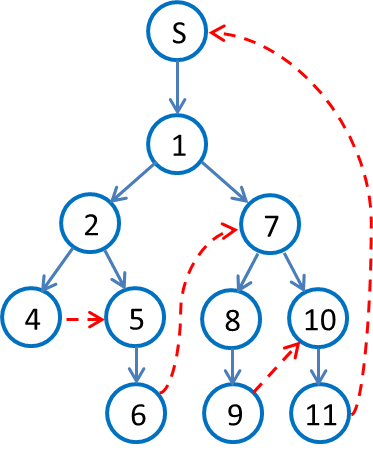}
     \label{fig:RemNode_2}}
     \hfil
     \subfigure[]{\includegraphics[height=35mm]{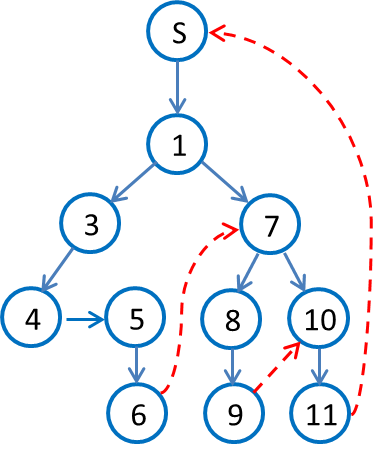}
    \label{fig:RemNode_3}}}
    \caption{An example illustrating the change in the topology of $G_1$, from Figure~\ref{fig:GenStruct_1}, due to the departure of (a) node 4, a leaf node in $T_1$, (b) node 3, a degree one node in $T_1$ and (c) node 2, a degree two node in $T_1$.}
    \label{fig:RemNode}
\end{figure}

\begin{figure}[t]
  \centerline{\subfigure[]{\includegraphics[height=40mm]{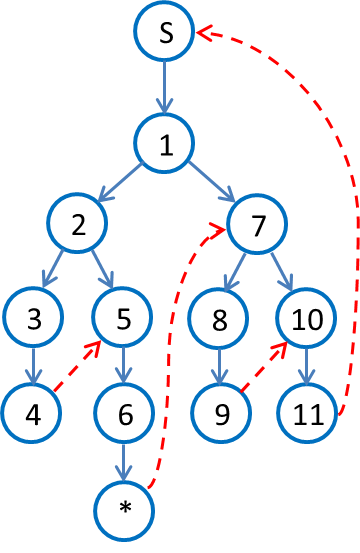}
     \label{fig:NodeArr_1}}
     \hfil
     \subfigure[]{\includegraphics[height=40mm]{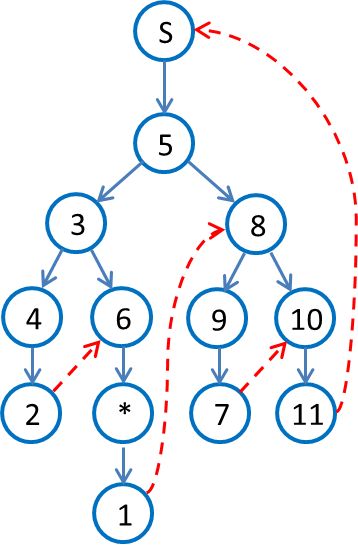}
     \label{fig:NodeArr_2}}
     \hfil
     \subfigure[]{\includegraphics[height=40mm]{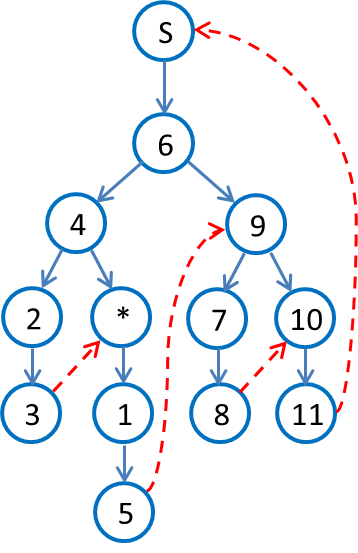}
    \label{fig:NodeArr_3}}}
    \caption{Topologies of $G_1$, $G_2$ and $G_3$ respectively, due to the arrival of a new node denoted by~*. The newly arrived peer is assumed to have contacted the node 6 in Figure~\ref{fig:GenStruct} initially.}
    \label{fig:NodeArr}
\end{figure}

Figures~\ref{fig:RemNode} and~\ref{fig:NodeArr} illustrate the node departure and arrival procedures respectively as discussed in section~\ref{sec: Peer Churn}.
\begin{algorithm}[H]
\caption{Edge update algorithm under parent departures for node $v$ in $G_i$}\label{Departure}
\begin{algorithmic}[1]
\Require labels of children in $T_i$; parent(s) of $v$ in $T_i$ and $U_i$ (if any) and their parents;
\Procedure{Departure}{$p$}
\If{$v$ is not a secondary child of $p$}
\State form the edges $(q,v)$ and $(r,v)$ as appropriate;\Comment{$q,r$: $p$'s parents in $T_i,U_i$ (if any)}
\State label($v$) $\gets$ label($v$)-1;
\State broadcast label update message $(i,\text{label}(v)-1,-1,t)$ along all edges of $\bigcup_{i=1}^m T_i$;
\Else
\State no action; $(r,v)$ becomes part of $T_i$ from $U_i$;\Comment{$r$: $v$'s parent in $U_i$}
\EndIf
\EndProcedure
\end{algorithmic}
\end{algorithm}

\begin{algorithm}[H]
\caption{Edge update algorithm under node arrivals for node $v$ in $G_i$}\label{Arrival}
\begin{algorithmic}[1]
\Require labels of children in $T_i$; parent $p$ of $v$ in $T_i$; parent $r$ of $v$ in $U_i$ if any; 
\Procedure{Arrival}{$q$}\Comment{$q$: node that requests to become parent}
\If{$q$ requests to insert itself between $p$ and $v$}
\If{degree($v$) $\neq 2$ in $T_i$}
\State break $(p,v)$ and form the edge $(q,v)$ in $T_i$;
\ElsIf{label($v_s$) $\neq \lceil (\textrm{label}(v)+l)/2 \rceil$}\Comment{$l$: control label received from $p$}
\State break $(p,v)$ and form the edge $(q,v)$ in $T_i$;
\Else{} reject request by $q$ to connect to $v$; retain old edge $(p,v)$; 
\EndIf
\ElsIf{$q$ requests to insert itself between $r$ and $v$}
\State break $(r,v)$ and form the edge $(q,v)$ in $U_i$;
\State give $p$'s address to $q$;\Comment{$q$ then requests to insert itself above $p$ in tree $T_{i+1}$} 
\EndIf
\State if label($v$) not consistent with label of parent(s) then update label and broadcast the label update message;\Comment{see procedure \Call{Departure}{}} 
\EndProcedure
\end{algorithmic}
\end{algorithm}

\subsection{~Active Balance} 
\begin{algorithm}[H]
\caption{Balancing of a degree two node $v$ in tree $T_1$}\label{ActiveBalance}
\begin{algorithmic}[1]
\Require labels of primary child $v_p$ and secondary child $v_s$ in $T_1$; parent $p$ of $v$ in $T_1$; 
\Procedure{ActiveBalance}{$l$}\Comment{$l$: control label received from parent in $T_1$}
\If{label($v_s$) $\neq \lceil (\textrm{label}(v)+l)/2 \rceil$ \textbf{and} $p$ is balanced}
\State break the edge $(v,v_s)$ in $T_1$;
\State find the node $u$ have the label $\lceil (\textrm{label}(v)+l)/2 \rceil$ in $G_i$;
\State form the edge $(v,u)$;    
\EndIf
\EndProcedure
\end{algorithmic}
\end{algorithm}

\subsection{~Induced Balance} 

In section~\ref{sec: Induced Balance} we have discussed a three step procedure by which peers in a balanced graph $G_i$ can induce its topology to a subsequent unbalanced $G_{i+1}$. The procedures \textsc{Request}, \textsc{Respond} and \textsc{Induce} are used to implement this. 
 
By the control label forwarding algorithm \textsc{Forward}, in algorithm~\ref{forward}, all degree one children of degree two parents receive the addresses of the secondary children of their corresponding degree two nodes in $G_i$ (such as node 3 receiving the address of node 5 in Figure~\ref{fig:GenStruct_1}). As such, in the first step, these nodes contact those secondary children. This is presented in Algorithm~\ref{Request} as the procedure \textsc{Request}. Since all secondary children of degree two nodes receive such a request, they can exchange this information among their neighbors in order to let the requesting degree one nodes know who their prospective parents will be in the subsequent graph. Indeed, procedure \textsc{Respond} in Algorithm~\ref{Respond} returns the address of the prospective parent and position (primary or secondary) in the subsequent tree to any requesting degree one node. In the last step, the degree one nodes request to insert themselves in the edge between the prospective parent and child in the subsequent tree (returned by \textsc{Respond}). This is done by procedure \textsc{Induce} in Algorithm~\ref{Induce}. Procedure \textsc{Request} takes 1 time slot, while procedures \textsc{Respond} and \textsc{Induce} takes 3 and 1 time slots respectively. As such, the entire operation occupies at most 5 time slots.  

\begin{algorithm}
\caption{Step 1 of algorithm for node $v$ to form the induced graph edges in $G_{i+1}$ from $G_i$}\label{Request}
\begin{algorithmic}[1]
\Require physical address of control label node $q$ received from parent $p$ of $v$ in $T_i$;
\Procedure{Request}{$q$}
\If{degree($v$) $= 1$ and degree($p$) $=2$ in $T_i$}
\If{label difference is $>2m-2$ in $T_i$}
\State create secondary edge
\ElsIf{$p$ is balanced} 
\State make an induced graph request to $q$; 
\EndIf
\ElsIf{degree($v$) $=2$ and ($l$ -- label($v_s$) $< m-1$ or label($v_s$) $-$ label($v_p$) $<m-1$)} 
\State break the edge to $v_s$; 
\EndIf
\EndProcedure
\end{algorithmic}
\end{algorithm}

\begin{figure}[!t]
  \centerline{\subfigure[]{\includegraphics[height=40mm]{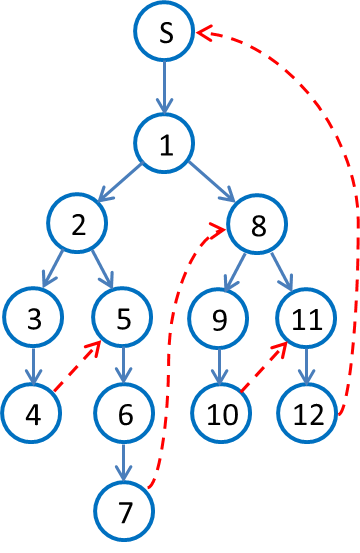}
     \label{fig:ReqArr_1}}
     \hfil
     \subfigure[]{\includegraphics[height=40mm]{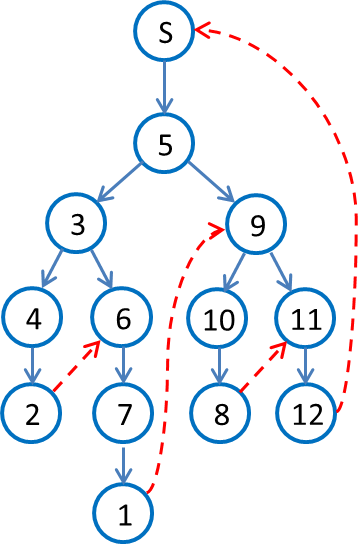}
     \label{fig:ReqArr_2}}
     \hfil
     \subfigure[]{\includegraphics[height=40mm]{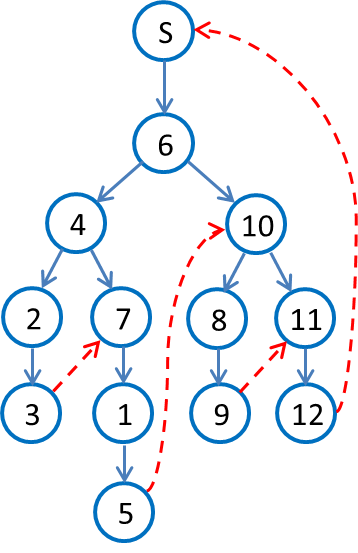}
    \label{fig:ReqArr_3}}}
    \caption{A possible steady state topology for $G_1$, $G_2$ and $G_3$ respectively for $n=12$. Notice that Property~\ref{ppt: steady state structure} holds.}
    \label{fig:ReqArr}
\end{figure}

\begin{algorithm}
\caption{Algorithm for node $v$ to respond to an induced graph request}\label{Respond}
\begin{algorithmic}[1]
\Require parent $p$ and children of $v$ in $T_i$;
\Procedure{Respond}{$q$}\Comment{$q$ makes the request to $v$}
\State send $q$'s address to $p$;
\If{degree($v$) $=2$} send address received from $v_s$ to $v_p$; 
\EndIf
\If{degree($v$) $=2$ and $v$ is a primary child} 
\State flag $\gets$ primary; 
\State send (address received from $p$, flag) to $v_s$; 
\ElsIf{degree($v$) $=2$ and $v$ is a secondary child} 
\State flag $\gets$ secondary; 
\State send ($q$'s address, flag) to $v_s$; 
\Else{} no action;
\EndIf
\State \textbf{return} (address, flag) received from $p$ to $q$;
\EndProcedure
\end{algorithmic}
\end{algorithm}

Note that \textsc{Induce} makes an insertion request to the child of its prospective parent in $G_{i+1}$. The \textsc{Arrival} routine running in the prospective child allows the insertion to take place only if the prospective child is unbalanced (if its a degree two node). This step ensures that if $G_{i+1}$ is already well balanced, then it is not perturbed by $G_i$. This has been illustrated in Figure~\ref{fig:ReqArr}, where node 7 in $G_3$ tries to break the edge $(S,1)$ in $G_1$ but is refused. Similarly, the procedure \textsc{Request} makes a request only if it cannot already be a degree two node in $G_i$. For example, if the degree one chain below a node is too long ($> 2m-2$), then it can form a secondary edge within $G_i$ itself. In addition, a request is made only if the parent of the requesting node is balanced, to ensure that if $G_i$ is  ill balanced then it does not propogate its structure to $G_{i+1}$. \textsc{Request} also makes sure that the degree one chains are at least $m-1$ nodes long, by breaking secondary edges if either of the subtrees contain less than $m-1$ nodes. 

\begin{algorithm}
\caption{Algorithm for node $v$ to form the induced graph edges in $G_{i+1}$ from $G_i$}\label{Induce}
\begin{algorithmic}[1]
\Require return values (address $q$, flag) of procedure \Call{Request}{}
\Procedure{Induce}{$q$, flag}
\State request the flag child of $q$ to insert itself between $q$ and flag child of $q$;
\EndProcedure
\end{algorithmic}
\end{algorithm}

In section~\ref{sec: Active Balance} we discussed a balancing procedure \textsc{ActiveBalance} in Algorithm~\ref{ActiveBalance} that actively tries to break and make new connections whenever the labels of the children hint at unbalanced subtrees. We remarked that such an active balancing algorithm is not necessary.  Consider the example of the departure of node 1 in Figure~\ref{fig:GenStruct}. This results in the topologies shown in Figure~\ref{fig:IndBal} for $G_1,G_2$ and $G_3$. Notice that $G_2$ and $G_3$ are balanced while $G_1$ is not. Hence if $G_3$ induces the topology $\hat{G}_1$ onto $G_1$, the graph $G_1$ can be balanced. This is illustrated in Figure~\ref{fig:IndBal_numbertwo} where node 4 first connects to the server and then forms its secondary edge to node 7. Hence by being conservative in breaking secondary edges, we can speeden the balancing process. We reiterate that the ability of nodes to reject an incoming insertion request by another node (\textsc{Arrival}), and the ability to make a degree two connection only if the inducing tree is balanced (\textsc{Request}) makes sure that an ill balanced graph cannot induce its structure onto a well balanced subsequent substream graph. Noting that it can take up to $m$ rounds for the cyclic inducing process to propagate from one graph to all the remaining, we have the following balancing algorithm.  

\begin{figure}[t]
  \centerline{\subfigure[]{\includegraphics[height=35mm]{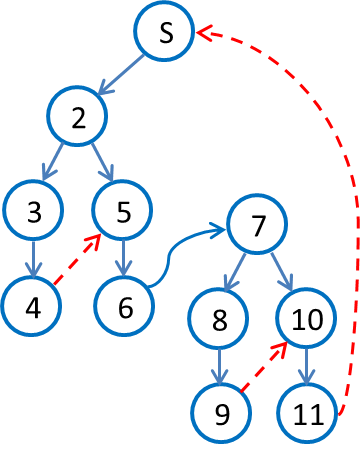}
     \label{fig:IndBal_1}}
     \hfil
     \subfigure[]{\includegraphics[height=35mm]{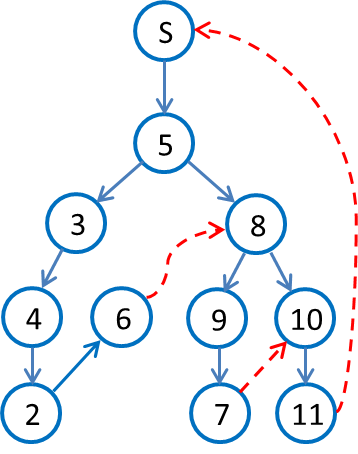}
     \label{fig:IndBal_2}}
     \hfil
     \subfigure[]{\includegraphics[height=35mm]{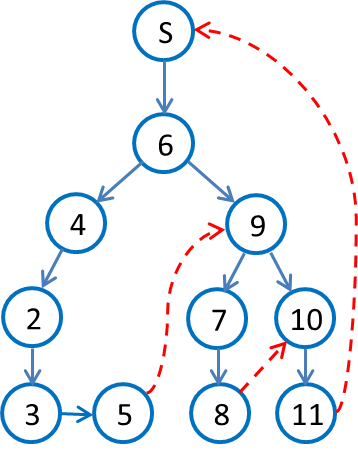}
    \label{fig:IndBal_3}}}
    \caption{Topology of (a) $G_1$, (b) $G_2$ and (c) $G_3$ resulting from the departure of node 1 from Figure~\ref{fig:GenStruct}. Notice that the edges $(3,6)$ in $G_2$ and $(4,5)$ in $G_3$ have been broken to preserve the minimum length of $2$ for the degree one chains.}
    \label{fig:IndBal}
\end{figure}

\begin{algorithm}
\caption{Balancing algorithm for node $v$ in $G_1$}\label{InducedBalance}
\begin{algorithmic}[1]
\Require labels of primary child $v_p$ and secondary child $v_s$ in $T_1$;
\Procedure{InducedBalance}{$l$}\Comment{$l$: control label received from parent in $T_1$}
\If{label($v_s$) $\neq \lceil (\textrm{label}(v)+l)/2 \rceil$ for $T_{\text{count}}\geq T_{\text{threshold}}$}\Comment{with $T_{\text{threshold}}=5m$}
\State run \Call{ActiveBalance}{$l$};
\ElsIf{label($v_s$) $\neq \lceil (\textrm{label}(v)+l)/2 \rceil$ for $T_{\text{count}} < T_{\text{threshold}}$}
\State $T_{\text{count}} \gets T_{\text{count}} + 1$ in the next time slot; 
\Else{} $T_{\text{count}} \gets 0$; 
\EndIf
\EndProcedure
\end{algorithmic}
\end{algorithm}

That is, we wait $T_{\text{threshold}}=5m$ rounds before breaking any secondary edge to form a new secondary edge. If after $5m$ time slots the labels are still incorrect, then the node breaks its secondary edge in tree $T_1$ as per \textsc{ActiveBalance}. This is because, if atleast one of the trees is balanced initially, then that tree can initiate the rearrangement cycle across all substream trees which takes at most $5m$ rounds. If none of the trees are balanced initially, then by initiating \textsc{ActiveBalance} tree $T_1$ gets balanced, which in turn causes the other trees to get balanced.

\subsection{~Multiple Departures} \label{sec: Multiple Departures}

\begin{figure}[t]
  \centerline{\subfigure[]{\includegraphics[height=35mm]{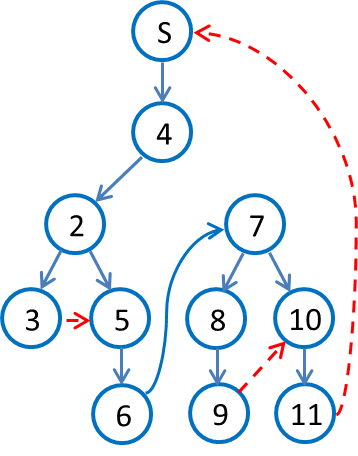}
     \label{fig:IndBal_4}}
     \hfil
     \subfigure[]{\includegraphics[height=35mm]{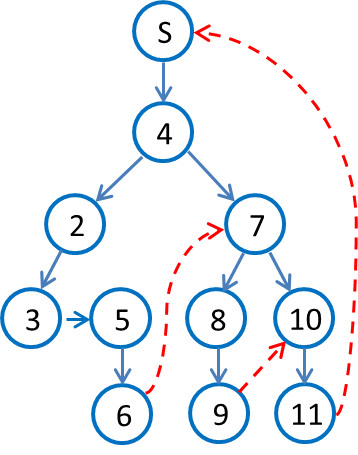}
     \label{fig:IndBal_5}}}
    \caption{The process of node 4 in $G_3$ of Figure~\ref{fig:IndBal} becoming a degree two node in $G_1$. }
    \label{fig:IndBal_numbertwo}
\end{figure}

So far we have been considering peer departures where only a single peer leaves the system at a time. Since we have not imposed any departure restrictions in our model, in general peers could depart in an arbitrary fashion including adversarial. Under such a scenario, the primary objective for the peers is to ensure connectivity in all the substream graphs. For example, in the network shown in Figure~\ref{fig:GenStruct} if the nodes 1,3,4,5 and 6 all leave at the same time then the node 2 is completely disconnected in all the graphs. As such, the performance is dictated by the amount of stored physical addresses, $M$, of the peers currently in the system. If none of the peers in the memory are available, then the node has to contact the server and re-enter the system as a new peer. 

\section{Proofs} \label{app: Gen tree degrees}

Consider the set of peers departing at time $t$. In the following proposition we show that the algorithm has the required resilience capability under churn. We remind that $M$ is the total amount of memory available in each peer, while $m$ is the number of substream graphs for a rate $R=m/(m+1)$. 

\begin{prop} \label{prop: small block}
If the peer departure blocks in each round are of size at most $K=1$, then for a memoy of $M=m$, the substream graphs $G_1,\ldots,G_m$ always satisfy Properties~\ref{ppt: connectivity} and~\ref{ppt: structure}. In general, the algorithm requires a memory of $M=Km$. 
\end{prop}
\begin{proof}
Note that the sizes of the sets of departing peers that are connected in $G_i$ is also bounded by $K$ for any $i$. In section~\ref{sec: Peer Churn} we have discussed the case of a single peer departure. Since the repairing procedure upon departure of any peer involves only the parent and children of the peer, arbitrary departures with bounded (by 1) departure block sizes can also be handled similarly. Now, if $G_1,\ldots,G_m$ satisfy Properties~\ref{ppt: connectivity} and~\ref{ppt: structure}, then it is easy to see that the repaired graphs resulting from a departure event also satisfy Properties~\ref{ppt: connectivity} and~\ref{ppt: structure}. Hence by induction we get the desired result for $K=1$. In general, connectivity of the peers is ensured if they have sufficient amount of memory to form new connections. For a memory of $M=Km$, every peer can know the address of the $K$ parent nodes above it in each $G_i$. Here we ignore the secondary edges and consider the $K$ parents in the resulting cycle graph (such as $1-2-3-4-\ldots-11$ in Figure~\ref{fig:GenStruct_1}). Now, if the peer departure blocks consist of at most $K$ nodes, then in the worst case a block consists of the $K$ parents of a peer. But in this case the peer can immediately restore connectivity by making a connection to the $K$-th parent above it. Hence, the proposition follows.    
\end{proof}
Since the distribution graphs $G_1,\ldots,G_m$ always satisfy Property~\ref{ppt: connectivity}, and their edges can support a rate of $1/(m+1)$ we conclude that peers suffer a loss of packets in at most one time slot required for the reconfiguration. 

In the following, we present the proof of the lower bound for the tree depth mentioned in Equation~\eqref{eq: delay expression}. 
\begin{prop} \label{prop: depth lower bound}
Any directed tree with $n$ nodes, and where $d^{(i)}$ fraction of the nodes have an out-degree of $i$ for $i=0,1,\ldots,l$, has a depth $D$ that is bounded as 
\begin{align}
D \geq \frac{d^{(1)}}{d^{(0)}} + \log_{l}\left(1+\sum_{k=2}^l nd^{(k)}(k-1)\right) - (l-2)\log_l\left(\frac{l!}{2}\right) .
\end{align}
\end{prop}
\begin{proof}
It is clear that the tree with the lowest depth, for a given $(d^{(0)},\ldots,d^{(l)})$, has the largest degree nodes on the very top followed by the second largest degree nodes and so on. Let us call a layer of nodes at a particular depth as an $i$-layer if the largest degree node present in that layer has the degree $i$. Further, let $d_i, i=1,\ldots,l$ denote the number of of the $i$-layers in the tree. Therefore,
\begin{align}
D = \sum_{i=0}^l d_i \label{eq: Depth expression}
\end{align}    
gives the depth of the tree. The proof proceeds by bounding the depth of each layer. The number of nodes in the topmost layer of the graph, layer $l$, can be bounded as
\begin{align}
1 + l + \ldots + l^{d_l - 2}  \leq nd^{(l)} \leq 1 + l + \ldots + l^{d_l - 1}
\end{align}
(If no such $d_l$ exists, then $d_l=0$). 
This yields
\begin{align}
\log_l(nd^{(l)}(l-1)+1) &\leq d_l, \\
l^{d_l} &\leq (nd^{(l)}(l-1)+1)l . \label{eq: layer l upper bound}
\end{align}
Now, in the second layer where there are nodes of degree $l-1$ (or possibly lesser), since $l^{d_l}$ constitutes an upper bound on the number of degree $l$ parents of degree $l-1$ nodes and $l^{d_l-1}(l-1)$ constiutes a lower bound, we must have
\begin{align}
l^{d_l-1}(l-1)(1+(l-1)+\ldots +(l-1)^{d_{l-1}-2})\leq nd^{(l-1)} \leq l^{d_l}(1+(l-1)+\ldots+(l-1)^{d_{l-1}-1}) \\
\Rightarrow l^{d_l-1}(1+(l-1)+\ldots +(l-1)^{d_{l-1}-2})\leq nd^{(l-1)} \leq l^{d_l}(1+(l-1)+\ldots+(l-1)^{d_{l-1}-1}). 
\end{align}
This yields
\begin{align}
\log_{l-1}\left( \frac{nd^{(l-1)}(l-2)}{l^{d_l} } +1\right)   &\leq d_{l-1}, \label{eq: layer l-1 lower bound}\\
l^{d_l}(l-1)^{d_{l-1}} &\leq (nd^{(l-1)}(l-2)+nd^{(l)}(l-1) + 1)(l)(l-1). \label{eq: layer l-1 upper bound}
\end{align}
Using Equation~\eqref{eq: layer l upper bound} in~\eqref{eq: layer l-1 lower bound} we have,
\begin{align}
\log_{l-1}\left( \frac{nd^{(l-1)}(l-2)}{(nd^{(l)}(l-1)+1)l}  +1\right)   \leq d_{l-1}. 
\end{align}
Similarly, we have in the $(l-2)$th layer, 
\begin{align}
l^{d_l-1}(l-1)(l-1)^{d_{l-1}-1}(l-2)(1+(l-2)+\ldots+(l-2)^{d_{l-2}-2})  & \leq nd^{(l-2)} \\ 
\Rightarrow l^{d_l-1}(l-1)^{d_{l-1}-1}(1+(l-2)+\ldots+(l-2)^{d_{l-2}-2})  & \leq nd^{(l-2)} \\ 
\text{and} \quad nd^{(l-2)} \leq l^{d_l}(l-1)^{d_{l-1}}(1+(l-2)+\ldots+(l-2)^{d_{l-2}-1}),   
\end{align}
yielding
\begin{align}
\log_{l-2}\left( \frac{nd^{(l-2)}(l-3)}{l^{d_l}(l-1)^{d_{l-1}} } + 1 \right) & \leq d_{l-2} \\
l^{d_l}(l-1)^{d_{l-1}}(l-2)^{d_{l-2}} & \leq (nd^{(l-2)}(l-3)+nd^{(l-1)}(l-2)+nd^{(l)}(l-1) + 1)(l)(l-1)(l-2).
\end{align}
Using Equation~\eqref{eq: layer l-1 upper bound} we have: 
\begin{align}
\log_{l-2}\left( \frac{nd^{(l-2)}(l-3)}{(nd^{(l-1)}(l-2)+nd^{(l)}(l-1) + 1)(l)(l-1)  } + 1 \right) & \leq d_{l-2}.
\end{align}
We continue this process for all the $i$-layers for $i\geq 2$. Finally, in the last layer the number of degree one chains is equal to the number of the leaves. As such, we must have 
\begin{align} 
\frac{nd^{(1)}}{nd^{(0)}} -1 & \leq d_1, \\
d_0 & = 1.
\end{align}
Therefore, from Equation~\eqref{eq: Depth expression} we have depth 
\begin{align}
D \geq \frac{d^{(1)}}{d^{(0)}} + \sum_{k=2}^l \log_{k}\left( 1+\frac{nd^{(k)}(k-1)}{(1+\sum_{k'=k+1}^l nd^{(k')}(k'-1))\prod_{k''=k+1}^l(k'')  } \right). \label{eq: depth lower bound}
\end{align}
Now, the second term in the right-hand side of Equation~\eqref{eq: depth lower bound}, denoted by $T$, can be lower bounded as
\begin{align}
T &\geq \sum_{k=2}^l \log_{l}\left( 1+\frac{nd^{(k)}(k-1)}{(1+\sum_{k'=k+1}^l nd^{(k')}(k'-1))\prod_{k''=k+1}^l(k'')  } \right) \\
&= \log_{l}\prod_{k=2}^l\left( 1+\frac{nd^{(k)}(k-1)}{(1+\sum_{k'=k+1}^l nd^{(k')}(k'-1))\prod_{k''=k+1}^l(k'')  } \right) \\
&= \log_{l}\prod_{k=2}^l\left( \frac{(1+\sum_{k'=k+1}^l nd^{(k')}(k'-1))\prod_{k''=k+1}^l(k'')+nd^{(k)}(k-1)}{(1+\sum_{k'=k+1}^l nd^{(k')}(k'-1))\prod_{k''=k+1}^l(k'')  } \right) \\
&\geq \log_{l}\prod_{k=2}^l\left( \frac{(1+\sum_{k'=k+1}^l nd^{(k')}(k'-1))+nd^{(k)}(k-1)}{(1+\sum_{k'=k+1}^l nd^{(k')}(k'-1))\prod_{k''=k+1}^l(k'')  } \right) \\
&= \log_{l}\left(\left(1+\sum_{k'=2}^l nd^{(k')}(k'-1)\right)\prod_{k=2}^l\left( \frac{1}{\prod_{k''=k+1}^l(k'')  } \right)\right) \\
&\geq \log_{l}\left(1+\sum_{k'=2}^l nd^{(k')}(k'-1)\right) - (l-2)\log_l\left(\frac{l!}{2}\right)
\end{align}
thus proving the claim.
\end{proof}
We now present the proofs of Equations~\eqref{eq: delay lower bound deterministic} and~\eqref{eq: d zero upper bound} from section~\ref{sec: Gen tree degrees}. For the sake of completeness we have presented the full-proof of Theorem~\ref{thm: order optimal delay}. 
\begin{proof}[Proof of Theorem~\ref{thm: order optimal delay}]
Without loss of generality let us consider $T$ trees with the $i$th tree carrying a rate of $r_i$. This is justified because if the flow is granular we can associate a shortest path tree with each of the substreams. The full topology of the multicast streams itself can be bigger than the trees above. The trees simply correspond to the routes by which the packets arrive earliest from the source to the peers. Let $d_i^{(j)}$ denote the fraction of nodes having an out-degree of $j$ in tree $i$. Clearly, 
\begin{align}
d_i^{(0)}+d_i^{(1)}+\ldots+d_i^{(l)} = 1, \quad \forall i=1,\ldots,T. \label{aeq: frac sum}
\end{align}
Since any tree with $n$ nodes has $n-1$ edges, we have
\begin{align}
n(d_i^{(1)}+2d_i^{(2)}+\ldots+(l-1)d_i^{(l-1)}+ld_i^{(l)}) = n-1, \quad \forall i=1,\ldots,T, \notag \\ 
\Rightarrow d_i^{(1)}+2d_i^{(2)}+\ldots+(l-1)d_i^{(l-1)}+ld_i^{(l)} = 1 - \frac{1}{n}, \quad \forall i=1,\ldots,T. \label{aeq: sum degree}
\end{align} 
Now, every degree $i$ node for $i\geq 2$ needs atleast $i-1$ redundant edges because of the capacity requirement of the theorem. As such, the cumulative node capacity constraint becomes
\begin{align}
\sum_{i=1}^T (n-1)r_i + n(d_i^{(2)}+2d_i^{(3)}+\ldots+(l-1)d_i^{(l)})r_i \leq n \notag \\ 
\Rightarrow \sum_{i=1}^T \left( 1 - \frac{1}{n} + d_i^{(2)}+2d_i^{(3)}+\ldots+(l-1)d_i^{(l)} \right)r_i \leq 1. \label{aeq: tolerance}
\end{align}
The proof essentially obtains a lower bound for the expression in Equation~\eqref{eq: delay expression} based on above Equations~\eqref{aeq: frac sum},~\eqref{aeq: sum degree} and~\eqref{aeq: tolerance}.  Subtracting Equation~\eqref{aeq: frac sum} from~\eqref{aeq: sum degree} gives 
\begin{align}
d_i^{(2)}+2d_i^{(3)}+\ldots+(l-1)d_i^{(l)} = d_i^{(0)} - \frac{1}{n}, \quad \forall i=1,\ldots,T. \label{eq: frac sum and sum degree}
\end{align}
From the above, we have 
\begin{align}
d_i^{(j)} \leq \frac{1}{j-1}\left(d_i^{(0)} - \frac{1}{n}\right),
\end{align}
and combined with Equation~\eqref{aeq: frac sum} we get
\begin{align}
1  &\leq d_i^{(0)} + d_i^{(1)} + \left(d_i^{(0)} - \frac{1}{n}\right)\left(1 + \frac{1}{2} +\ldots +\frac{1}{l-1} \right) \notag \\
 &\leq d_i^{(0)} + d_i^{(1)} + \left(d_i^{(0)} - \frac{1}{n}\right) (\log_e(l-1)+1) \notag \\
\Rightarrow d_i^{(1)}  &\geq 1 - d_i^{(0)} (\log_e(l-1)+2) + \frac{1}{n}(\log_e(l-1)+1) \notag  \\
\Rightarrow 
\frac{d_{i}^{(1)}}{d_{i}^{(0)}} &\geq \frac{1}{d_{i}^{(0)}} - (\log_e(l-1)+2) + \frac{1}{nd_{i}^{(0)}}(\log_e(l-1)+1). \label{aeq: first term lower bound}
\end{align}
Also, the second term in the delay lower bound in Equation~\eqref{eq: delay expression} becomes
\begin{align}
\log_{l}\left(1+\sum_{k=2}^l nd^{(k)}(k-1)\right)  = \log_{l}\left(1+nd^{(0)}_{i}\right). \label{aeq: second term lower bound}
\end{align}
As such, using Equations~\eqref{aeq: first term lower bound},~\eqref{aeq: second term lower bound} and~\eqref{eq: delay expression} the delay for the $i$-th tree $D_i$ can now be lower bounded as
\begin{align}
D_i \geq \frac{1}{d_{i}^{(0)}} - (\log_e(l-1)+2)  + \log_{l}\left(1+nd^{(0)}_{i}\right) - (l-2)\log_l\left(\frac{l!}{2}\right),\quad\forall i=1,\ldots,T. \label{aeq: delay lower bound deterministic}
\end{align}
The derivative of the right-hand side above in Equation~\eqref{aeq: delay lower bound deterministic} with respect to $d_i^{(0)}$ is given by 
\begin{align}
-\frac{1}{(d_i^{(0)})^2}+\frac{n}{(1+nd_i^{(0)})\log l}
\end{align}
which is strictly negative in $0<d_i^{(0)}<1$. As such, the minima in the right-hand side of Equation~\eqref{aeq: delay lower bound deterministic} is achieved by the largest achievable $d_i^{(0)}$. 
Now, using Equations~\eqref{aeq: frac sum} and~\eqref{aeq: sum degree} in~\eqref{aeq: tolerance} we get
\begin{align}
\sum_{i=1}^T \left(1-\frac{2}{n} + d_i^{(0)} \right) r_i \leq 1,  \notag \\
\Rightarrow \min_i d_i^{(0)} \leq \frac{1}{R} - 1 + \frac{2}{n}.  \label{aeq: i dont know}
\end{align}
Letting $i^* = \mathrm{arg\,min}~ d_i^{(0)}$, the overall delay for the system can be bounded by the delay of the $i^*$-th tree. Hence, substituting Equation~\eqref{aeq: i dont know} in~\eqref{aeq: delay lower bound deterministic} we have 
\begin{align}
D &\geq \frac{1}{d_{i^*}^{(0)}} - (\log_e(l-1)+2)  + \log_{l}\left(1+nd^{(0)}_{i^*}\right) - (l-2)\log_l\left(\frac{l!}{2}\right) \notag  \\
&\geq \frac{1}{\frac{1}{R} - 1 + \frac{2}{n}} + \log_{l}\left(1+n\left( \frac{1}{R} - 1 + \frac{2}{n} \right)\right)- (l-2)\log_l\left(\frac{l!}{2}\right) - \log_e(l-1) - 2 \notag \\
&\geq \log_ln+\frac{R}{2(1-R)}+\log_{l}\left(\frac{2(1-R)}{R}\right)- (l-2)\log_l\left(\frac{l!}{2}\right)   - \log_e(l-1)-2
\end{align}
for $n\geq 3R/(1-R)$. For $l=\Delta$ and a node capacity of $C$ (rather than 1) replacing $R$ by $R/C$, we get the desired theorem. Hence we can conclude that the steady state delay in our algorithm, Theorem~\eqref{thm: steady state delay}, is order optimal for the class of algorithms satisfying the conditions of Theorem~\ref{thm: order optimal delay}. 
\end{proof}

\section{All-Cast} \label{app: All-Cast}

In the all-cast scenario, each peer in the system can have an independent data stream for broadcasting to all the other peers. The symmetry of the distribution topology that we constructed for the broadcast problem in sections~\ref{sec: Overview} -- \ref{sec: Proofs} allows us to reuse the topology for all-cast. Let us assume the node capacities of the peers proportionally scale as the number of streaming sources in the system. For example, if there are $k$ independent broadcasts then we will assume that the peers can support a total upload rate of $k$. The rate of each independent stream and its substreams are the same as in the original algorithm. In the single source broadcast graph the edges carrying the data streams were directed. However, since the edges of the P2P network have been assumed to be undirected in our model in section~\ref{sec: model}, we allow data transfer to happen both directions between any of the neighbours in the substream graphs $G_i,i=1,\ldots,m$. As such, let $\bar{G}_i$ denote the undirected version of the directed graph $G_i$ for all $i$. Then, for any source node $v$, the problem is to find a rooted (at the source), low depth, directed spanning tree (where the edges point away from the source) in $\bar{G}_i$ subject to the node capacity constraints on the peers. One way to ensure the capacity constraints is to find a route in $\bar{G}_i$, for each independent broadcast stream, such that the out-degree for the peers is the same as in $G_i$ for all $i$. 

This can be done as follows. Consider the single source broadcast algorithm for a rate $R=m/(m+1)$. This results in the contruction of $m$ substream graphs $G_1,\ldots,G_m$. Let $v$ be any peer sourcing a data stream. For substream $i$, if $v$ is a degree two node in $T_i$, then $v$ sends the substream to its primary child and parent in $T_i$. Otherwise, if $v$ is of degree one or zero, it sends the stream to its child in $T_i$ or $U_i$ respectively. Now, for any node $u$ that is receiving a substream from its neighbor, if $u$ receives it from any of the neighbors in $T_i$, it forwards the substream to the other neighbors in $T_i$. If $u$ is a leaf-node in $T_i$ receiving messages from its parent, then $u$ forwards the messages to it child in $U_i$. On the other hand, if $u$ receives the substream from its neighbor in $U_i$, then if $u$ is of degree two in $T_i$ and $u$'s parent in $T_i$ has not yet received the stream then it forwards it to its parent and primary child in $T_i$. If $u$'s parent has received the stream, then it forwards to its children in $T_i$. If $u$ has degree one, it forwards the message to its child in $T_i$. This algorithm has been presented in Algorithm~\ref{AllCast} and illustrated in Figure~\ref{fig:AllCast}. In all of the above operations, the amount of upload done by the nodes for each substream of each independent stream is the same as in the original algorithm. 

\begin{figure}[t]
  \centerline{\subfigure[]{\includegraphics[height=35mm]{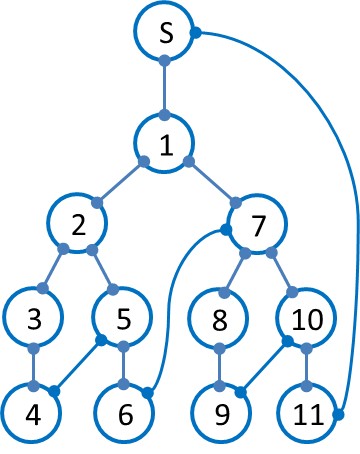}
     \label{fig:AllCast_1}}
     \hfil
     \subfigure[]{\includegraphics[height=35mm]{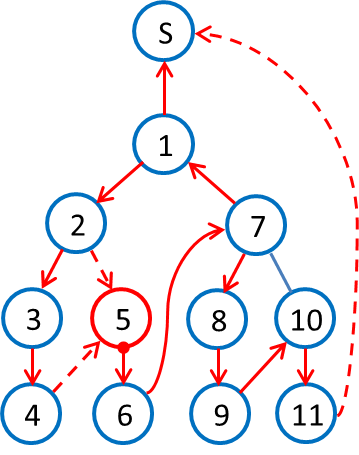}
     \label{fig:AllCast_2}}
     \hfil
     \subfigure[]{\includegraphics[height=35mm]{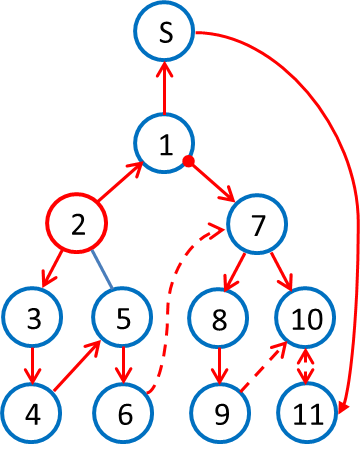}
    \label{fig:AllCast_3}}}
    \caption{(a) The undirected network corresponding to the directed graph in Figure~\ref{fig:GenStruct_1}, (b) route taken by the stream if node 5 is the source and (c) route taken by the stream for node 2 as the source.}
    \label{fig:AllCast}
\end{figure}

\begin{algorithm}
\caption{Algorithm for node $v$ to source / forward an all-cast message in $G_i$ }\label{AllCast}
\begin{algorithmic}[0]
\Require all neighbors in $G_i$;
\Procedure{AllCast}{} \\
 Case: $v$ is the source of msg
\If{degree($v$) $=2$}
\State send msg to primary child and parent in $T_i$;  
\Else 
\State send msg to the child in $T_i$ or $U_i$; 
\EndIf \\
 Case: $v$ receives msg from a neighbor
\If{degree($v$) $=2$} 
\State if msg was received from a neighbor in $T_i$, forward msg to the other two neighbors in $T_i$; 
\State if msg was received from a neighbor in $U_i$ and the parent in $T_i$ has not yet received msg, then forward msg to primary child and the parent in $T_i$; else forward msg to primary and secondary children in $T_i$; 
\Else 
\State forward msg to the other neighbor in $T_i$ or $U_i$; 
\EndIf
\EndProcedure
\end{algorithmic}
\end{algorithm}

\begin{prop}
For a substream graph $G_i, i\in [m]$ satisfying Properties~\ref{ppt: connectivity} and~\ref{ppt: structure}, procedure \textsc{AllCast} ensures that all the nodes in $G_i$ receive the substream for any source of the broadcast.   
\end{prop}
\begin{proof}
Suppose a node or a set of nodes do not receive the stream. Then we can always find a node in that set whose parent in $T_i$ or $U_i$ have received the stream. It cannot happen that any parent in $U_i$ received the stream and the child did not. As such, the only possiblity is that the node is a secondary child of a degree two node in $T_i$. But in this case, the primary child of the parent has received the stream. Now, since the node is a secondary child it also has a parent in $U_i$ which has not received the stream (otherwise the node would have recived it). There exists a directed path comprising of only primary edges, and tolerance edges from the primary child to the secondary child. Going backwards along this path implies the primary child did not get the stream, which is a contradiction.   
\end{proof}

\begin{prop}
In the steady state with $n$ nodes, the delay of any of the streams in the all-cast is bounded by 
\begin{align} 
D_\text{all-cast} \leq 2\log_2(n+1) +  \frac{8R}{1-R} + 2\log_2(1-R) - 8.
\end{align} 
\end{prop}
\begin{proof}
For any degree two node in $T_i$, let $\overline{T}_i$ and $\underline{T}_i$ denote the tree above and below the node respectively. From Equation~\eqref{eq: steady state log depth}, the depth of the degree two portion of $T_i$ in the steady state is bounded by $d\leq \log_2\left(\frac{n+1}{m+1}\right)$ and by Property~\ref{ppt: steady state structure}, the length of the degree one chains are at most $2m-2$. Now, it takes at most $2d + 2m-2$ delay for the stream to reach all the nodes in $\overline{T}_i$. For $\underline{T}_i$ it takes at most $2d + 2(2m-2)$ delay. Therefore, it takes at most $2d + 2(2m-2)$ delay for any degree node that is a source. Now, if any degree one node is the source, then it takes at most $2(2m-2)$ rounds to reach a degree two node. From there on it behaves as if the degree two node is the source and hence takes at most $2d + 2(2m-2)$ delay. Hence, the net delay bounded by $2d + 4(2m-2)\leq 2\log_2(n+1) +  \frac{8R}{1-R} + 2\log_2(1-R) - 8$ as required. 
\end{proof}

\section{Heterogeneous Capacities} \label{app: Heterogeneous nodes}

In a setting where peers have heterogeneous upload capacities, it is easily seen that the maximum possible streaming capacity is equal to the sum upload capacity of the peers divided by the number of peers~\cite{munwebwei06a}. Likewise, a (weak) lower bound for the maximum delay is $\Omega(\log n)$ under constant node degree bounds. Intuitively it seems possible to be able to trade one quantity for the other, such as rate for delay etc. However, precisely characterizing the rate-delay-continuity tradeoff (analogous to section~\ref{sec: RDT tradeoff}) in the heterogeneous case remains an important future direction. 

In this section we contribute to the above question, by considering the ``low-rate low-delay'' regime (at zero-tolerance, $\tau=0$). Without loss of generality, let the peers have an upload capacity greater than or equal to $1$. Then this regime corresponds to streaming at a rate of $R\leq 1$. The other direction is the ``high-rate high-delay'' regime, and corresponds to transmission at an optimal (or near-optimal) rate as discussed above. We have not considered this direction, and leave it for future work. The low rate regime is similar in spirit to~\cite{zhu2013tree}, where a few dedicated high capacity peers (or servers) assist in faster data dissemination by being located in the top of the distribution trees.     

The key idea here is to cluster together nodes of similar upload capacities and run the original algorithm on the clusters separately. Let us first consider the homogeneous case, as before, but with multiple source nodes providing the data stream instead of just one in each of the $G_i$'s. 

\subsection{~Multiple Source Nodes} \label{sec: multiple source nodes}

\begin{figure}[t]
  \centerline{\subfigure[]{\includegraphics[height=35mm]{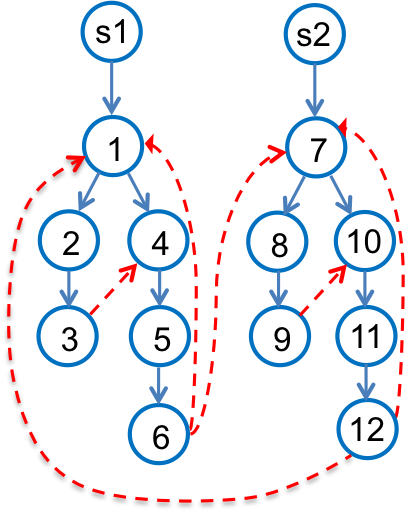}
     \label{fig:MultSource_1}}
     \hfil
     \subfigure[]{\includegraphics[height=35mm]{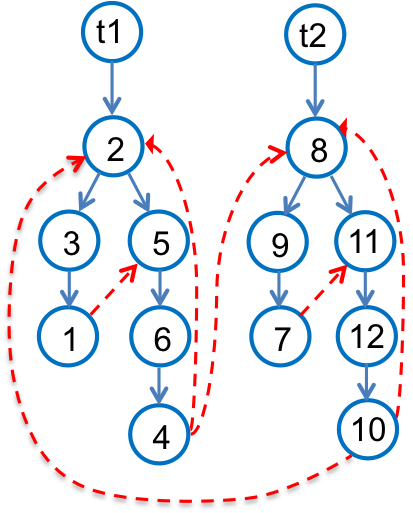}
     \label{fig:MultSource_2}}
     \hfil
     \subfigure[]{\includegraphics[height=35mm]{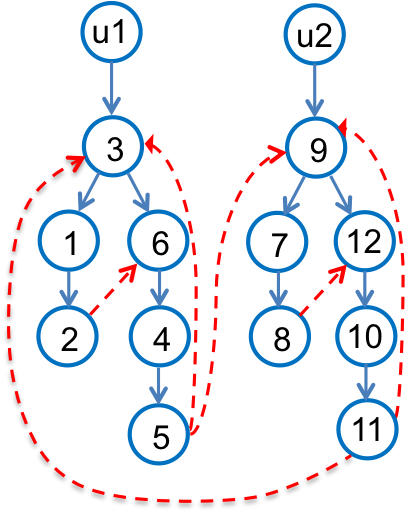}
    \label{fig:MultSource_3}}}
    \caption{Example of a cluster $C_1$ comprising of 12 nodes with 2 source nodes. (a), (b) and (c) show the three substream graphs for $R=3/4$.}
    \label{fig:MultSource}
\end{figure}

In our algorithm for the streaming model of section~\ref{sec: model}, for a rate of $R=m/(m+1)$, the server provided the stream to a single node in the substream graphs $G_1,\ldots,G_m$. The receiving nodes have a label 1 in their respective substream labeling in the steady state (Figure~\ref{fig:GenStruct}). Now, let us suppose there are $k_1\geq1$ servers providing the substream for $G_1$. In this case, we can expect the topology to comprise of $k_1$ balanced graphs (satisfying Properties~\ref{ppt: connectivity} and~\ref{ppt: structure}) in the steady state. The steady state topology of $G_1,G_2$ and $G_3$ with 12 peers and 2 sources has been illustrated in Figure~\ref{fig:MultSource}. We have remarked that the end nodes of the substream graphs (such as nodes 10, 11 in Figure~\ref{fig:GenStruct}) are atypical and do not use their full upload capacity across $G_1,\ldots,G_m$. However, while dealing with multiple servers, the extra capacity in the end nodes can be used to connect to both the root node the parent tree and the subsequent tree. The label information forwarded along these edges can be used for balancing the parent tree and also for ensuring that the trees are of similar size. 

\subsection{~Streaming in Clusters}
Let $C_1$ and $C_{1+}$ denote the set of peers with upload capacity 1 and strictly larger than 1 respectively. We assume that whenever new peers arrive they can obtain the address of an arbitrary peer in their respective clusters. Then, for a rate of $R=m/(m+1)$, we let the peers in $C_{1+}$ form and maintain the distribution graphs exactly as before in section~\ref{sec: Algorithm}. However, since the peers have an upload capacity strictly larger than 1, this does not use all of their capacity. The remaining capacity available in those nodes are used as sources for the peers in the cluster $C_1$ as in the previous section~\ref{sec: multiple source nodes}. One way to do this is to let the degree one children of degree two nodes use all of their extra capacity for sourcing that substream to $C_1$. In Figure~\ref{fig:ToC1} we have illustrated this for a cluster $C_{1+}$ where every peer has an upload capacity of $5/4$ for a rate $R=3/4$. Since each substream is of rate $1/4$, the peers in $C_{1+}$ can support up to 5 outgoing edges. While 4 edges are used for the construction of the substream graphs, the remaining edges (shown by dotted lines in the Figure) are used as source nodes for the lower capacity cluster $C_1$. For example, if $C_1$ is as in Figure~\ref{fig:MultSource}, then peers $2,4$ in Figure~\ref{fig:ToC1} can be the sources $s1,s2$ in Figure~\ref{fig:MultSource} corresponding to the first substream and so on. 

Now, peer churn can happen in terms of peer arrivals and departures in both $C_1$ and $C_{1+}$. Since the distribution graphs for the peers in $C_{1+}$ are exactly as before, peer churn can also be handled similarly. However, for the peers in $C_1$, churn in $C_{1+}$ translates as dynamics in the number of substream sources. In addition, they have to handle the peer churn happening within their cluster. The latter is handled as in the homogeneous scenario (section~\ref{sec: Algorithm}) since the distribution graphs of $C_1$ have Properties~\ref{ppt: connectivity} and~\ref{ppt: structure} for small block departures (Proposition~\ref{prop: small block}), while the connectivity property ensures that the peers continue to receive the stream even when some of the substream sources from $C_{1+}$ leave the system.      

\begin{figure}[t]
  \centerline{\subfigure[]{\includegraphics[height=40mm]{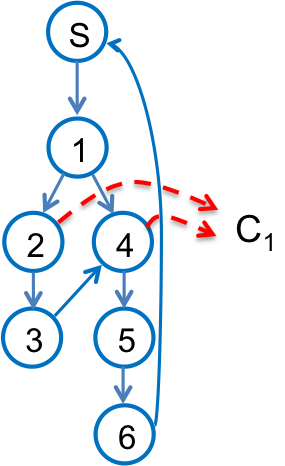}
     \label{fig:ToC1_1}}
     \hfil
     \subfigure[]{\includegraphics[height=40mm]{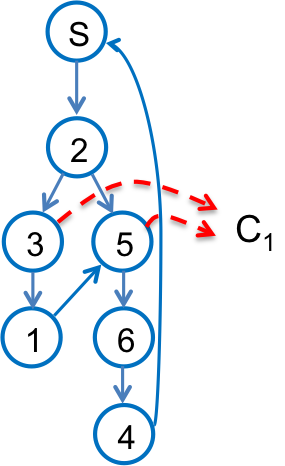}
     \label{fig:ToC1_2}}
     \hfil
     \subfigure[]{\includegraphics[height=40mm]{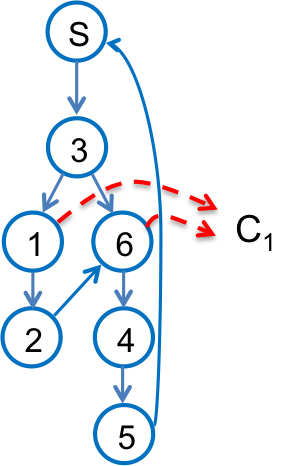}
    \label{fig:ToC1_3}}}
    \caption{A cluster $C_{1+}$ with 6 peers. The nodes with dotted-red edges act as source nodes for $C_1$.}
    \label{fig:ToC1}
\end{figure}

As with churn management, balance has to be achieved in both the clusters. For the peers in $C_{1+}$ using the algorithm of section~\ref{sec: Algorithm} this is automatically guaranteed. However, for the nodes in $C_1$, the delay is minimized if the trees corresponding to each source in each substream are of similar size. Note that every root node of a substream in $C_1$ can know the size of its tree since it receives edges from its end node and the end node of the previous tree. By exchange of this information among the root nodes, they can direct the source nodes to make connections such that the trees are of similar sizes. For example, in Figure~\ref{fig:MultSource_1} the end node 6 forwards the label information to the root nodes 1 and 7. Similarly node 12 forwards the label information to the two root nodes. As such, by taking the difference of the received label, the root nodes 1 and 7 can know the size of their respective subtrees as 6. The root nodes can also know the size of the neighboring tree by exchanging this informtion using the end nodes 6 and 12. As such, if there are $k$ sources, each subtree root can know the size of its own subtree and the neighboring subtrees. The root nodes can then use this information to direct the source nodes for that substream to find new root nodes such that the $k$ subtrees are approximately equal in size. Then, balance within the trees can be achieved as in the homogeneous case.

Thus, our algorithm can easily be extended to cover the ``low-rate low-delay'' regime of heterogeneous networks. Details and analysis are left to the full-paper. 
\end{document}